\documentclass[USenglish,cleveref,autoref,thm-restate]{lipics-v2021}
\pdfoutput=1

\hideLIPIcs
\nolinenumbers

\usepackage{style}
\usepackage{shortcuts}

\title{A Faster Algorithm for Constrained Correlation Clustering}

\author{Nick Fischer}{INSAIT, Sofia University ``St. Kliment Ohridski'', Bulgaria}{nick.fischer@weizmann.ac.il}{}{Partially funded by the Ministry of Education and Science of Bulgaria’s support for INSAIT, Sofia University ``St. Kliment Ohridski'' as part of the Bulgarian National Roadmap for Research Infrastructure. Parts of this work were done while the author was at Saarland University.}
\author{Evangelos Kipouridis}{Max Planck Institute for Informatics, Saarland Informatics Campus, Saarbr{\"u}cken, Germany}{kipouridis@mpi-inf.mpg.de}{https://orcid.org/0000-0002-5830-5830}{}
\author{Jonas Klausen}{BARC, University of Copenhagen, Denmark}{jokl@di.ku.dk}{https://orcid.org/0000-0002-7403-417X}{}
\author{Mikkel Thorup}{BARC, University of Copenhagen, Denmark}{mikkel2thorup@gmail.com}{https://orcid.org/0000-0001-5237-1709}{}
\authorrunning{N. Fischer, E. Kipouridis, J. Klausen, and M. Thorup}

\Copyright{Nick Fischer, Evangelos Kipouridis, Jonas Klausen, and Mikkel Thorup}

\ccsdesc{Theory of computation~Facility location and clustering}
\keywords{Clustering, Constrained Correlation Clustering, Approximation}

\relatedversion{}
\funding{
\emph{Jonas Klausen} and \emph{Mikkel Thorup} are part of BARC, Basic Algorithms Research Copenhagen, supported by VILLUM Foundation grants 16582 and 54451.
This work is part
of the project TIPEA that has received funding from the European Research Council (ERC) under the European
Unions Horizon 2020 research and innovation programme (grant agreement No. 850979).}

\acknowledgements{We thank Lorenzo Beretta for his valuable suggestions on weighted sampling.}

\EventEditors{Olaf Beyersdorff, Micha\l{} Pilipczuk, Elaine Pimentel, and Nguyen Kim Thang}
\EventNoEds{4}
\EventLongTitle{42nd International Symposium on Theoretical Aspects of Computer Science (STACS 2025)}
\EventShortTitle{STACS 2025}
\EventAcronym{STACS}
\EventYear{2025}
\EventDate{March 4--7, 2025}
\EventLocation{Jena, Germany}
\EventLogo{}
\SeriesVolume{327}
\ArticleNo{39}

\begin{document}
\maketitle

\begin{abstract}
In the Correlation Clustering problem we are given $n$ nodes, and a preference for each pair of nodes indicating whether we prefer the two endpoints to be in the same cluster or not.
The output is a clustering inducing the minimum number of violated preferences.
In certain cases, however, the preference between some pairs may be too important to be violated. The constrained version of this problem specifies pairs of nodes that \emph{must} be in the same cluster as well as pairs that \emph{must not} be in the same cluster (hard constraints).
The output clustering has to satisfy all hard constraints while minimizing the number of violated preferences. 

Constrained Correlation Clustering is APX-Hard and has been approximated within a factor $3$ by van Zuylen \emph{et al.} [SODA '07]. Their algorithm is based on rounding an LP with $\Theta(n^3)$ constraints, resulting in an $\Omega(n^{3\omega})$ running time.
In this work, using a more combinatorial approach, we show how to approximate this problem significantly faster at the cost of a slightly weaker approximation factor.
In particular, our algorithm runs in $\widetilde{O}(n^3)$ time (notice that the input size is $\Theta(n^2)$) and approximates Constrained Correlation Clustering within a factor $16$.

To achieve our result we need properties guaranteed by a particular influential algorithm for (unconstrained) Correlation Clustering, the \textsc{CC-PIVOT} algorithm.
This algorithm chooses a \emph{pivot} node $u$, creates a cluster containing $u$ and all its preferred nodes, and recursively solves the rest of the problem.
It is known that selecting pivots at random gives a $3$-approximation.
As a byproduct of our work, we provide a derandomization of the \textsc{CC-PIVOT} algorithm that still achieves the $3$-approximation; furthermore, we show that there exist instances where no ordering of the pivots can give a $(3-\varepsilon)$-approximation, for any constant $\varepsilon$.

Finally, we introduce a node-weighted version of Correlation Clustering, which can be approximated within factor $3$ using our insights on Constrained Correlation Clustering. As the general weighted version of Correlation Clustering would require a major breakthrough to approximate within a factor $o(\log{n})$, Node-Weighted Correlation Clustering may be a practical alternative.
\end{abstract}

\clearpage
\setcounter{page}{1}
\section{Introduction}
Clustering is a fundamental task related to unsupervised learning, with many applications in machine learning and data mining. The goal of clustering is to partition a set of nodes into disjoint clusters, such that (ideally) all nodes within a cluster are similar, and nodes in different clusters are dissimilar. As no single definition best captures this abstract goal, a lot of different clustering objectives have been suggested.

\emph{Correlation Clustering} is one of the most well studied such formulations for a multitude of reasons: Its definition is simple and natural, it does not need the number of clusters to be part of the input, and it has found success in many applications. Some few examples include automated labeling~\cite{autoLabelAgrawal,autoLabelChakrabarti}, clustering ensembles~\cite{clusteringEnsembles}, community detection~\cite{commDet, communityVeldt}, disambiguation tasks~\cite{disambiguation}, duplicate detection~\cite{duplicate} and image segmentation~\cite{image1,image2}.

In \CC{} we are given a graph $G=(V,E)$, and the output is a partition (clustering) $C=\{C_1, \ldots, C_k\}$ of the vertex set~$V$. We refer to the sets $C_i$ of $C$ as \emph{clusters}. The goal is to minimize the number of edges between different clusters plus the number of non-edges inside of clusters. More formally, the goal is to minimize $|E \symdiff E_C|$, the cardinality of the symmetric difference between $E$ and $E_C$, where we define~$E_C=\bigcup_{i=1}^{k} \binom{C_i}{2}$.
In other words, the goal is to transform the input graph into a collection of cliques with the minimal number of edge insertions and deletions.
An alternative description used by some authors is that we are given a \emph{complete} graph where the edges are labeled either ``$+$'' (corresponding to the edges in our graph $G$) or ``$-$'' (corresponding to the non-edges in our graph $G$).

The problem is typically motivated as follows: Suppose that the input graph models the relationships between different entities which shall be grouped. An edge describes that we prefer its two endpoints to be clustered together, whereas a non-edge describes that we prefer them to be separated. In this formulation the cost of a correlation clustering is the number of violated preferences.

\subsection{Previous Results}
\CC{} was initially introduced by Bansal, Blum, and Chawla~\cite{Bansal}, who proved that it is NP-Hard, and provided a deterministic constant-factor approximation, the constant being larger than 15,000. Subsequent improvements were based on rounding the natural LP: Charikar, Guruswami and Wirt gave a deterministic $4$-approximation~\cite{cutRadius}, Ailon, Charikar and Newman gave a randomized $2.5$-approximation and proved that the problem is APX-Hard~\cite{pivoting}, while a deterministic $2.06$-approximation was given by Chawla, Makarychev, Schramm and Yaroslavtsev~\cite{NearOptimal2}. The last result is near-optimal among algorithms rounding the natural LP, as its integrality gap is at least $2$. In a breakthrough result by Cohen-Addad, Lee and Newman \cite{sub2} a $(1.994+\epsilon)$-approximation using the Sherali-Adams relaxation broke the $2$ barrier. It was later improved to $1.73+\epsilon$ \cite{apx173} by Cohen-Addad, Lee, Li and Newman, and even to $1.437$ by Cao \emph{et al.}~\cite{clusterLP}. There is also a combinatorial $1.847$-approximation (Cohen-Addad \emph{et al.}~\cite{combinatorialMikkel}).

Given the importance of Correlation Clustering, research does not only focus on improving its approximation factor. Another important goal is efficient running times without big sacrifices on the approximation factor. As the natural LP has $\Theta(n^3)$ constraints, using a state-of-the-art LP solver requires time $\Omega(n^{3\omega})=\Omega(n^{7.113})$.
In order to achieve efficient running times, an algorithm thus has to avoid solving the LP using an all-purpose LP-solver, or the even more expensive Sherali-Adams relaxation; such algorithms are usually called \emph{combinatorial} algorithms\footnote{On a more informal note, combinatorial algorithms are often not only faster, but also provide deeper insights on a problem, compared to LP-based ones.}.
Examples of such a direction can be seen in~\cite{pivoting} where, along with their LP-based $2.5$-approximation, the authors also design a combinatorial $3$-approximation (the \textsc{CC-PIVOT} algorithm); despite its worse approximation, it enjoys the benefit of being faster. Similarly, much later than the $2.06$-approximation~\cite{NearOptimal2}, Veldt devised a faster combinatorial $6$-approximation and a $4$-approximation solving a less expensive LP~\cite{deterministic}.

Another important direction is the design of \emph{deterministic} algorithms. For example, \cite{pivoting} posed as an open question the derandomization of \textsc{CC-PIVOT}.
The question was (partially) answered affirmatively by~\cite{deterministicPivoting}.
Deterministic algorithms were also explicitly pursued in~\cite{deterministic}, and are a significant part of the technical contribution of~\cite{NearOptimal2}.

Correlation Clustering has also been studied in different settings such as parameterized algorithms~\cite{fpt}, sublinear and streaming algorithms~\cite{sublinear, dynamicStream, BehMany, BehSingle, dynamicStream, MakarySingle}, massively parallel computation~(MPC) algorithms~\cite{parallel, sub3parallel}, and differentially private algorithms~\cite{differentialPrivacy}.

\subparagraph*{PIVOT.}
The \textsc{CC-PIVOT} algorithm \cite{pivoting} is a very influential algorithm for Correlation Clustering.
It provides a $3$-approximation and is arguably the simplest constant factor approximation algorithm for Correlation Clustering. It simply selects a node uniformly at random, and creates a cluster $\mathcal{C}$ with this node and its neighbors in the (remaining) input graph. It then removes $\mathcal{C}$'s nodes and recurses on the remaining graph.
Due to its simplicity, \textsc{CC-PIVOT} has inspired several other algorithms, such as algorithms for Correlation Clustering in the streaming model \cite{sublinear, dynamicStream, BehMany, BehSingle, dynamicStream, MakarySingle} and algorithms for the more general Fitting Utrametrics problem \cite{charikarTree, l0}.

One can define a meta-algorithm based on the above, where we do not necessarily pick the pivots uniformly at random.
Throughout this paper, we use the term \textsc{PIVOT} algorithm to refer to an instantiation of the (Meta-)Algorithm~\ref{alg:pivot}\footnote{This is not to be confused with the more general pivoting paradigm for \CC{} algorithms. In that design paradigm, the cluster we create for each pivot is not necessarily the full set of remaining nodes with which the pivot prefers to be clustered, but can be decided in any other way (e.g. randomly, based on a probability distribution related to an LP or more general hierarchies such as the Sherali-Adams hierarchy).}.
Obviously \textsc{CC-PIVOT} is an instantiation of \textsc{PIVOT}, where the pivots are selected uniformly at random.

\begin{algorithm2e}[thb]
  \nonl\procedure{\Pivot{$G = (V, E)$}}{
    $C \gets \emptyset$\;
    \While{$V \neq \emptyset$}{
      Pick a pivot node $u$\;
      \tcc{an instantiation of \Pivot{} only needs to specify how the pivot is selected in each iteration} 
      Add a cluster containing $u$ and all its neighbors to $C$\;
      Remove $u$, its neighbors and all their incident edges from $G$\;
    }
    \Return $C$
  }

\caption{The \textsc{PIVOT} meta-algorithm. \textsc{CC-PIVOT} is an instantiation of \textsc{PIVOT} where pivots are selected uniformly at random.} \label{alg:pivot}
\end{algorithm2e}

The paper that introduced \textsc{CC-PIVOT} \cite{pivoting} posed as an open question the derandomization of the algorithm.
The question was partially answered in the affirmative by~\cite{deterministicPivoting}.
Unfortunately there are two drawbacks with this algorithm.
First, it requires solving the natural LP, which makes its running time equal to the pre-existing (better) $2.5$-approximation.
Second, this algorithm does not only derandomize the order in which pivots are selected, but also decides the cluster of each pivot based on an auxiliary graph (dictated by the LP) rather than based on the original graph.
Therefore it is not an instantiation of \textsc{PIVOT}.

\subparagraph*{Weighted Correlation Clustering.}
In the weighted version of \CC{}, we are also given a weight for each preference. The final cost is then the sum of weights of the violated preferences. An $O(\log{n})$-approximation for weighted \CC{} is known by Demaine, Emanuel, Fiat and Immorlica~\cite{weighted}. In the same paper they show that the problem is equivalent to the Multicut problem, meaning that an $o(\log{n})$-approximation would require a major breakthrough. As efficiently approximating the general weighted version seems out of reach, research has focused on special cases for which constant-factor approximations are possible~\cite{global, local}.

\subparagraph*{\CCC{}.}
\CCC{} is an interesting variant of \CC{} capturing the idea of critical pairs of nodes. To address these situations, \CCC{} introduces hard constraints in addition to the pairwise preferences. A clustering is valid if it satisfies all hard constraints, and the goal is to find a valid clustering of minimal cost. We can phrase \CCC{} as a weighted instance of Correlation Clustering: Simply give infinite weight to pairs associated with a hard constraint and weight $1$ to all other pairs.

To the best of our knowledge, the only known solution to \CCC{} is given in the work of van Zuylen and Williamson who designed a deterministic $3$-approximation \cite{deterministicPivoting}.
The running time of this algorithm is $O(n^{3\omega})$, where $\omega<2.3719$ is the matrix-multiplication exponent.
Using the current best bound for $\omega$, this is $\Omega(n^{7.113})$.

\subsection{Our Contribution}
Our main result is the following theorem.
It improves the $\Omega(n^{7.113})$ running time of the state-of-the-art algorithm for \CCC{} while still providing a constant (but larger than $3$) approximation factor\footnote{We write $\widetilde\Order(T)$ to suppress polylogarithmic factors, i.e., $\widetilde\Order(T) = T (\log T)^{\Order(1)}$.}.

\begin{theorem}[Constrained Correlation Clustering] \label{thm:CCCCombinatorial}
There is a deterministic algorithm for \CCHC{} computing a $16$-approximation in time~$\widetilde O(n^3)$.
\end{theorem}

We first show how to obtain this result, but with a randomized algorithm that holds with high probability, instead of a deterministic one.
In order to do so, we perform a (deterministic) preprocessing step and then use the \textsc{CC-PIVOT} algorithm.
Of course \textsc{CC-PIVOT} alone, without the preprocessing step, would not output a clustering respecting the hard constraints.
Its properties however (and more generally the properties of \textsc{PIVOT} algorithms) are crucial; we are not aware of any other algorithm that we could use instead and still satisfy all the hard constraints of \CCC{} after our preprocessing step.

To obtain our deterministic algorithm we derandomize the \textsc{CC-PIVOT} algorithm.

\begin{theorem}[Deterministic \textrm{PIVOT}] \label{thm:CCCombinatorial}
There are the following deterministic PIVOT algorithms for Correlation Clustering:
\smallskip
\begin{itemize}
    \item A combinatorial $(3+\epsilon)$-approximation, for any constant $\epsilon>0$, in time $\widetilde O(n^3)$.
    \item A non-combinatorial $3$-approximation in time $\widetilde O(n^5)$.
\end{itemize}
\end{theorem}

We note that the final approximation of our algorithm for \CCC{} depends on the approximation of the applied \textsc{PIVOT} algorithm.
If it was possible to select the order of the pivots in a way that guarantees a better approximation, this would immediately improve the approximation of our \CCC{} algorithm.
For this reason, we study lower bounds for \textrm{PIVOT};
currently, we know of instances for which selecting the pivots at random doesn't give a better-than-$3$-approximation in expectation~\cite{pivoting}; however, for these particular instances there \emph{does} exist a way to choose the pivots that gives better approximations.
Ideally, we want a lower bound applying for any order of the pivots (such as the lower bound for the generalized \textsc{PIVOT} solving the Ultrametric Violation Distance problem in \cite{l0}).
We show that our algorithm is optimal, as there exist instances where no ordering of the pivots will yield a better-than-$3$-approximation.

\begin{theorem}[\textrm{PIVOT} Lower Bound] \label{thm:pivot-lower-bound}
There is no constant $\epsilon>0$ for which there exists a PIVOT algorithm for \CC{} with approximation factor $3-\epsilon$.
\end{theorem}

We also introduce the \NWCC{} problem, which is related to (but incomparable, due to their asymmetric assignment of weights) a family of problems introduced in \cite{communityVeldt}. As weighted Correlation Clustering is equivalent to Multicut, improving over the current~$\Theta(\log{n})$-approximation seems out of reach. The advantage of our alternative type of weighted Correlation Clustering is that it is natural and approximable within a constant factor.

In \NWCC{} we assign weights to the nodes, rather than to pairs of nodes. Violating the preference between nodes $u, v$ with weights $\omega_u$ and $\omega_v$ incurs cost $\omega_u \cdot \omega_v$.
We provide three algorithms computing (almost-)3-approximations for \NWCC{}:

\begin{theorem}[Node-Weighted Correlation Clustering, Deterministic] \label{thm:node-weighted-deterministic}
There are the following deterministic algorithms for Node-Weighted Correlation Clustering:
\begin{itemize}
    \item A combinatorial $(3 + \epsilon)$-approximation, for any constant $\epsilon > 0$, in time $\widetilde\Order(n^3)$.
    \item A non-combinatorial $3$-approximation in time $\Order(n^{7.116})$.
\end{itemize}
\end{theorem}

\begin{theorem}[Node-Weighted Correlation Clustering, Randomized] \label{thm:node-weighted-randomized}
There is a randomized combinatorial algorithm for Node-Weighted Correlation Clustering computing an expected $3$-approximation in time $\Order(n + m)$ with high probability $1 - 1/\poly(n)$.
\end{theorem}

\subsection{Overview of Our Techniques}
\subparagraph*{\CCC{}.}
We obtain a faster algorithm for \CCC{} by
\begin{enumerate}
	\item modifying the input graph using a subroutine aware of the hard-constraints, and
	\item applying a \textsc{PIVOT} algorithm on this modified graph.
\end{enumerate}
In fact, no matter what \textsc{PIVOT} algorithm is used, the output clustering respects all hard constraints when the algorithm is applied on the modified graph.

To motivate this two-step procedure, we note that inputs exist where \emph{no} PIVOT algorithm, if applied to the unmodified graph, would respect the hard constraints.
One such example is the cycle on four vertices, with two vertex-disjoint edges made into hard constraints.

The solution of \cite{deterministicPivoting} is similar to ours, as it also modifies the graph before applying a \CC{} algorithm.
However, both their initial modification and the following \CC{} algorithm require solving the standard LP, which is expensive ($\Omega(n^{7.113})$ time).
In our case both steps are implemented with deterministic and combinatorial algorithms which brings the running time down to $\widetilde{O}(n^3)$.

For the first step, our algorithm carefully modifies the input graph so that on one hand the optimal cost is not significantly changed, and on the other hand any \textsc{PIVOT} algorithm on the transformed graph returns a clustering that respects all hard constraints.
For the second step, we use a deterministic combinatorial \textsc{PIVOT} algorithm.

Concerning the effect of modifying the graph, roughly speaking we get that the final approximation factor is $\myBetaPlusOne{}\cdot \alpha + 3$, where $\alpha$ is the approximation factor of the \textsc{PIVOT} algorithm we use. Plugging in $\alpha = 3+\epsilon$ from \cref{thm:CCCombinatorial} we get the first combinatorial constant-factor approximation for \CCC{} in $\widetilde{O}(n^3)$ time.

\subparagraph*{\NWCC{}.}
We generalize the deterministic combinatorial techniques from before to the Node-Weighted Correlation Clustering problem. In addition, we also provide a very efficient randomized algorithm for the problem. It relies on a weighted random sampling technique.

One way to view the algorithm is to reduce Node-Weighted Correlation Clustering to an instance of \CCC{}, with the caveat that the new instance's size depends on the weights (and can thus even be exponential).
Each node $u$ is replaced by a set of nodes of size related to $u$'s weight and these nodes have constraints forcing them to be in the same cluster.

We show that we can simulate a simple randomized \textsc{PIVOT} algorithm on that instance, where instead of sampling uniformly at random, we sample with probabilities proportional to the weights.
Assuming polynomial weights, we can achieve this in linear time.
To do so, we design an efficient data structure supporting such sampling and removal of elements.

It is easy to implement such a data structure using any balanced binary search tree, but the time for constructing it and applying all operations would be $O(n\log{n})$.
Using a non-trivial combination of the Alias Method~\cite{Walker74,alias} and Rejection Sampling, we achieve a linear bound.

Due to space constraints the presentation of our algorithms for \NWCC{} is deferred to \cref{sec:nodeWeightedAppendix}.

\subparagraph*{Deterministic PIVOT algorithms.}
Our algorithms are based on a simple framework by van Zuylen and Williamson~\cite{deterministicPivoting}. In this framework we assign a nonnegative ``charge'' to each pair of nodes. Using these charges, a \textsc{PIVOT} algorithm decides which pivot to choose next. The approximation factor depends on the total charge (as compared with the cost of an optimal clustering), and the minimum charge assigned to any bad triplet~(an induced subgraph~$K_{1,2}$).

The reason why these bad triplets play an important role is that for any bad triplet, any clustering needs to pay at least $1$.
To see this, let $uvw$ be a bad triplet with $uv$ being the only missing edge.
For a clustering to pay $0$, it must be the case that both $uw$ and $vw$ are together. However, this would imply that $uv$ are also together although they prefer not to.

Our combinatorial $(3+\epsilon)$-approximation uses the multiplicative weights update method, which can be intuitively described as follows: We start with a tiny charge on all pairs. Then we repeatedly find a bad triplet $uvw$ with currently minimal charge (more precisely: for which the sum of the charges of $uv, vw, wu$ is minimal), and scale the involved charges by $1+\epsilon$. One can prove that this eventually results in an almost-optimal distribution of charges, up to rescaling. 

For this purpose it suffices to show that the total assigned charge is not large compared to the cost of the optimal correlation clustering. We do so by observing that our algorithm $(1+\epsilon)$-approximates the covering LP of \cref{fig:lps}, which we refer to as the \emph{charging~LP}.

Our faster deterministic non-combinatorial algorithm solves the charging LP using an LP solver tailored to covering LPs~\cite{Allen-ZhuO19,WangRM16}.
An improved solver for covering LPs would directly improve the running time of this algorithm.

\begin{figure}[h]
\caption{The primal and dual LP relaxations for Correlation Clustering, which we refer to as the \emph{charging LP}. $T(G)$ is the set of all bad triplets in $G$.} \label{fig:lps}
\vskip -1.5ex\rule{\linewidth}{.5pt}
\begin{mini}|s|<b>{}{\sum_{uv \in \binom V2} x_{uv}}{}{}
    \addConstraint{x_{uv} + x_{vw} + x_{wu}}{\geq 1\quad}{\forall uvw \in T(G)}
    \addConstraint{x_{uv}}{\geq 0\quad}{\forall uv \in \textstyle\binom{V}{2}}
\end{mini}
\rule{\linewidth}{.5pt}
\begin{maxi}|s|<b>{}{\sum_{uvw \in T(G)} y_{uvw}}{}{}
    \addConstraint{\sum_{w : uvw \in T(G)} y_{uvw}}{\leq 1\quad}{\forall uv \in \textstyle \binom{V}{2}}
    \addConstraint{y_{uvw}}{\geq 0\quad}{\forall uvw \in T(G)}
\end{maxi}
\rule{\linewidth}{.5pt}
\end{figure}

\subparagraph*{Lower Bound.}
Our lower bound is obtained by taking a complete graph $K_n$ for some even number of vertices $n$, and removing a perfect matching. Each vertex in the resulting graph is adjacent to all but one other vertex and so \emph{any} \textsc{PIVOT} algorithm will partition the vertices into a large cluster of $n-1$ vertices and a singleton cluster. A non \textsc{PIVOT} algorithm, however, is free to create just a single cluster of size $n$, at much lower cost. The ratio between these solutions tends to $3$ with increasing $n$.

We note that in \cite{pivoting} the authors proved that \textsc{CC-PIVOT}'s analysis is tight. That is, its expected approximation factor is not better than $3$. However, their lower bound construction (a complete graph $K_n$ minus one edge) only works for \textsc{CC-PIVOT}, not for \textsc{PIVOT} algorithms in general.

\subsection{Open Problems}
We finally raise some open questions.
\begin{enumerate}
    \item Can we improve the approximation factor of Constrained Correlation Clustering from $16$ to $3$ while keeping the running time at $\widetilde O(n^3)$?
    \item We measure the performance of a \textsc{PIVOT} algorithm by comparing it to the best correlation clustering obtained by \emph{any} algorithm. But as \cref{thm:pivot-lower-bound} proves, there is no \textsc{PIVOT} algorithm with an approximation factor better than $3$. If we instead compare the output to the best correlation clustering obtained by a \emph{PIVOT algorithm}, can we get better guarantees (perhaps even an exact algorithm in polynomial time)?
    \item In the Node-Weighted Correlation Clustering problem, we studied the natural objective of minimizing the total cost $\omega_v \cdot \omega_u$ of all violated preferences $uv$. Are there specific applications of this problem? Can we achieve similar for other cost functions such as~\makebox{$\omega_v + \omega_u$}?
\end{enumerate}

\section{Preliminaries}
We denote the set $\{1,\ldots,n\}$ by $[n]$. We denote all subsets of size $k$ of a set $A$ by $\binom{A}{k}$. The symmetric difference between two sets $A,B$ is denoted by $A\symdiff B$. We write~$\poly(n) = n^{\Order(1)}$ and~$\widetilde\Order(n) = n (\log n)^{\Order(1)}$.

In this paper all graphs $G=(V,E)$ are undirected and unweighted. We typically set~$n = |V|$ and $m = |E|$. For two disjoint subsets $U_1,U_2 \subseteq V$, we denote the set of edges with one endpoint in~$U_1$ and the other in~$U_2$ by $E(U_1,U_2)$. The subgraph of $G$ induced by vertex-set $U_1$ is denoted by $G[U_1]$. For vertices~$u, v, w$ we often abbreviate the (unordered) set $\set{u, v}$ by $uv$ and similarly write $uvw$ for $\set{u, v, w}$.
We say that $uvw$ is a \emph{bad triplet} in~$G$ if the induced subgraph $G[uvw]$ contains exactly two edges (i.e., is isomorphic to $K_{1, 2}$). Let $T(G)$ denote the set of bad triplets in $G$.
We say that the edge set \(E_C\) of a clustering $C=\set{C_1,\ldots, C_k}$ of $V$ is the set of pairs with both endpoints in the same set in $C$. More formally, \(E_C = \bigcup_{i=1}^k \binom{C_i}{2}\).

We now formally define the problems of interest.

\begin{definition}[Correlation Clustering]
Given a graph $G = (V, E)$, output a clustering $C=\set{C_1, \ldots, C_k}$ of $V$ with edge set $E_C$ minimizing $|E \symdiff E_C|$.
\end{definition}

An algorithm for \CC{} is said to be a \textsc{PIVOT} algorithm if it is an instantiation of \cref{alg:pivot} (\cpageref{alg:pivot}).
That is, an algorithm which, based on some criterion, picks an unclustered node \(u\) (the \emph{pivot}), creates a cluster containing \(u\) and its unclustered neighbors in \((V, E)\), and repeats the process until all nodes are clustered.
In particular, the algorithm may not modify the graph in other ways before choosing a pivot.

The constrained version of Correlation Clustering is defined as follows.

\begin{definition}[Constrained Correlation Clustering]
Given a graph $G = (V, E)$, a set of friendly pairs~$F\subseteq \binom{V}{2}$ and a set of hostile pairs $H\subseteq \binom{V}{2}$, compute a clustering $C=\set{C_1,\ldots, C_k}$ of~$V$ with edge set $E_C$ such that no pair $uv \in F$ has $u, v$ in different clusters and no pair~$uv \in H$ has $u, v$ in the same cluster. The clustering $C$ shall minimize $|E \symdiff E_C|$.
\end{definition}

We also introduce \NWCC{}, a new related problem that may be of independent interest.

\begin{definition}[Node-Weighted Correlation Clustering]
Given a graph $G = (V, E)$ and positive weights $\set{\omega_u}_{u \in V}$ on the nodes, compute a clustering $C=\set{C_1, \ldots, C_k}$ of $V$ with edge set~$E_C$ minimizing
\begin{equation*}
    \sum_{uv \in E\symdiff E_C} \omega_u \cdot \omega_v \, .
\end{equation*}
\end{definition}

For simplicity, we assume that the weights are bounded by $\poly(n)$, and thereby fit into a constant number of word RAM cells of size $w=\Theta(\log n)$.
We remark that our randomized algorithm would be a polynomial (but not linear) time one if we allowed the weights to be of exponential size.

The Node-Weighted Correlation Clustering problem clearly generalizes Correlation Clustering since we pay $w(u)\cdot w(v)$ (instead of $1$) for each pair $uv$ violating a preference.
\section{Combinatorial Algorithms for \CCHC{}} \label{sec:CCHC}
Let us fix the following notation: A connected component in~$(V,F)$ is a \emph{supernode}. The set of supernodes partitions $V$ and is denoted by $SN$. Given a node $u$, we let $s(u)$ be the unique supernode containing $u$. Two supernodes $U,W$ are \emph{hostile} if there exists a hostile pair $uw$ with $u\in U, w\in W$. Two supernodes $U,W$ are \emph{connected} if~$|E(U,W)|\ge 1$. Two supernodes $U,W$ are \emph{$\beta$-connected} if $|E(U,W)| \ge \beta \cdot |U| \cdot |W|$.

The first step of our combinatorial approach is to transform the graph $G$ into a more manageable form $G'$, see procedure \Transform of \cref{alg:main}.
The high-level idea is that in $G'$:
\smallskip
\begin{enumerate}
    \item \label{prop:friendly} If $uv$ is a friendly pair, then $u$ and $v$ are connected and have the same neighborhood.
    \item \label{prop:hostile} If $uv$ is a hostile pair, then $u$ and $v$ are not connected and have no common neighbor.
    \item \label{prop:approximation} An $O(1)$-approximation of the $G'$ instance is also an $O(1)$-approximation of the $G$ instance.
\end{enumerate}
\smallskip
As was already noticed in~\cite{deterministicPivoting}, Properties~\ref{prop:friendly} and~\ref{prop:hostile} imply that a \textsc{PIVOT} algorithm on~$G'$ gives a clustering satisfying the hard constraints.
Along with Property~\ref{prop:approximation} and our deterministic combinatorial \textsc{PIVOT} algorithm for Correlation Clustering in \cref{thm:CCCombinatorial}, we prove \cref{thm:CCCCombinatorial}.
Properties~\ref{prop:friendly} and \ref{prop:hostile} (related to correctness) and the running time ($\widetilde{O}(n^3)$) of our algorithm are relatively straightforward to prove.
Due to space constraints, their proofs can be found in \cref{sec:analysisCCC}.
In this section we instead focus on the most technically challenging part, the approximation guarantee.

\begin{figure}[t]
  \centering
  \subfloat[][The original graph. The set of friendly pairs is $F=\{\{1,2\}, \{2,3\}, \{4,5\}, \{6,7\}\}$, and the only hostile pair in $H$ is $\{2,5\}$.]{\includegraphics[page=1, width=.4\textwidth]{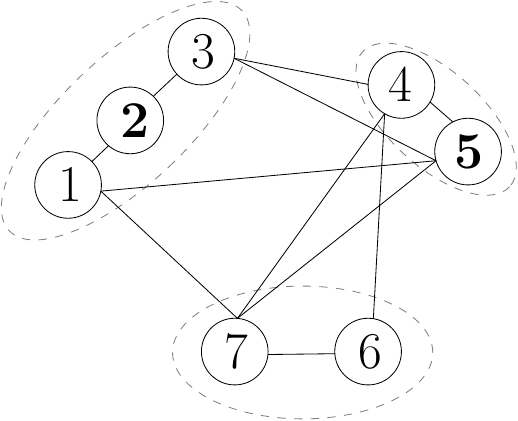}}\quad\quad
  \subfloat[][Line~\ref{line:connectSupernode} introduces edge $\{1,3\}$, and Line~\ref{line:disconnectSupernode} disconnects the supernodes containing $2$ and $5$.]{\includegraphics[page=2, width=.4\textwidth]{transform.pdf}}\\
  \subfloat[][Line~\ref{line:drop2} removes the pair of edges $\{1,7\}$ and $\{4,6\}$ because $1,4$ are in hostile supernodes while $6,7$ are in the same supernode.]{\includegraphics[page=3, width=.4\textwidth]{transform.pdf}}\quad\quad
  \subfloat[][Line~\ref{line:round} introduces all edges connecting supernodes $\{4,5\}$ and $\{6,7\}$ because there were enough edges between them already.]{\includegraphics[page=4, width=.4\textwidth]{transform.pdf}}
  \caption{Illustrates an application of \textsc{TRANSFORM(G,F,H)} (Algorithm~\ref{alg:main}). In the transformed graph, for any two supernodes~$U_1,U_2$, either all pairs with an endpoint in $U_1$ and an endpoint in $U_2$ share an edge, or none of them do. Furthermore, all pairs within a supernode are connected and no hostile supernodes are connected.}
  \label{fig:sub1}
\end{figure}

Our algorithm works as follows (see also \Cref{fig:sub1}):
If some supernode is hostile to itself, then it outputs that no clustering satisfies the hard constraints. Else, starting from the edge set $E$, it adds all edges within each supernode. Then it drops all edges between hostile supernodes. Subsequently, it repeatedly detects hostile supernodes that are connected with the same supernode, and drops one edge from each such connection. Finally, for each $\beta$-connected pair of supernodes, it connects all their nodes if $\beta>\thrshld{}$, and disconnects them otherwise\footnote{The constant $\thrshld{}$ optimizes the approximation factor. The natural choice of $0.5$ would still give a constant approximation factor, albeit slightly worse.}.

From a high-level view, the first two modifications are directly related to the hard constraints: If $u_1,u_2$ are friendly and $u_2,u_3$ are friendly, then any valid clustering has $u_1,u_3$ in the same cluster, even if a preference discourages it. Similarly, if $u_1,u_2$ are friendly, $u_3,u_4$ are friendly, but $u_1,u_3$ are hostile, then any valid clustering has $u_2,u_4$ in different clusters, even if a preference discourages it. Our first two modifications simply make the preferences consistent with the hard constraints.

The third modification guarantees that hostile supernodes share no common neighbor.
A \textsc{PIVOT} algorithm will thus never put their nodes in the same cluster, as the hostility constraints require. Concerning the cost, notice that if hostile supernodes $U_1,U_2$ are connected with supernode $U_3$, then no valid clustering can put all three of them in the same cluster. Therefore we always need to pay either for the connections between $U_1$ and~$U_3$, or for the connections between $U_2$ and $U_3$. 

Finally, after the rounding step, for each pair of supernodes $U_1, U_2$, the edge set $E(U_1, U_2)$ is either empty or the full set of size $|U_1| \cdot |U_2|$. This ensures that a \textsc{PIVOT} algorithm always puts all nodes of a supernode in the same cluster, thus also obeying the friendliness constraints. Concerning the cost of the rounded instance, a case analysis shows that it is always within a constant factor of the cost of the instance before rounding.

\begin{algorithm2e}[!ht]
  \SetKwFunction{MAIN}{ConstrainedCluster}

  \setcounter{AlgoLine}{0}
  
  \nonl \procedure{\Transform{$G=(V,E), F, H$}}{
  Compute the connected components of $(V,F)$\;
  \medskip
  \tcp{Impossible iff some pair must both be and not be in the same cluster.}
  \lIf{$\exists U\in SN$ hostile to itself}{\Return $G'=(\emptyset, \emptyset)$}
  
  \medskip
  \tcp{Connect nodes in the same supernode.}
  $E_1 \gets E \cup \set{uv \in \binom{V}{2} \mid s(u)=s(v)}$ \label{line:connectSupernode}\;
  
  \medskip
  \tcp{Disconnect pairs in hostile supernodes.}
  $E_2 \gets E_1 \setminus \set{uv \in \binom{V}{2} \mid \text{$s(u)$ and $s(v)$ are hostile}}$ \label{line:disconnectSupernode}\;
  
  \medskip
  \tcp{While hostile supernodes $U_1,U_2$ are both connected with super-\\node~$U_3$, drop an edge between $U_1,U_3$ and an edge between $U_2,U_3$}
  $E_3 \gets E_2$\;
  \While{$\exists U_1,U_2,U_3\in \binom{SN}{3}$ such that $U_1,U_2$ are hostile and $\exists u_1\in U_1, u_2\in U_2, u_3\in U_3, u_3'\in U_3$ such that $u_1u_3\in E_3, u_2u_3' \in E_3$}{ \label{line:drop2}
    $E_3 \gets E_3 \setminus \set{u_1 u_3, u_2 u_3'}$
  }
  
  \medskip
  \tcp{Round connections between pairs of supernodes}
  $E_4 \gets E_3$ \; \label{line:afterDrop2}
  \ForEach{$\{U_1,U_2\} \in \binom{SN}{2}$} { \label{line:round}
    $E_{U_1,U_2} \gets \{ u_1u_2 \mid u_1\in U_1, u_2\in U_2\}$\;
    \lIf{$|E_{U_1,U_2} \cap E_4| > \thrshld{}|U_1| \cdot |U_2|$} {$E_4 \gets E_4 \cup E_{U_1,U_2}$}
    \lElse {$E_4\gets E_4 \setminus E_{U_1,U_2}$}
  }

  \Return $G'=(V,E_4)$
  }

  \nonl\;
  \nonl \procedure{\MAIN{$G=(V,E), F, H$}}{
    $G' \gets \Transform{G=(V,E), F, H}$\;
    \lIf {$G'=(\emptyset,\emptyset)$}{\Return ``Impossible''}
    \Return $\textsc{PIVOT}(G')$\;
  }

\caption{The procedure $\textsc{ConstrainedCluster}$ is given a graph $G = (V, E)$ describing the preferences, a set of friendly pairs $F$ and a set of hostile pairs $H$. It creates a new graph $G'$ using the procedure $\textsc{Transform}$ and uses any $\textsc{PIVOT}$ algorithm on $G'$ to return a clustering.}
\label{alg:main}
\end{algorithm2e}

Formally, let $E'$ be the edge set of the transformed graph~$G'$, let \(E_3\) be the edge set at Line~\ref{line:afterDrop2} of \cref{alg:main} (exactly before the rounding step), $\OPT$ be the edge set of an optimal clustering for $E$ satisfying the hard constraints described by $F$ and $H$, $\OPT'$ be the edge set of an optimal clustering for the preferences defined by $E'$, and $E_C$ be the edge set of the clustering returned by our algorithm. Finally, let $\alpha$ be the approximation factor of the $\textsc{PIVOT}$ algorithm used.
\begin{lemma} \label{lem:supernode}
Given an instance $(V,E,F,H)$ of \CCHC, if two nodes $u_1,u_2$ are in the same supernode, then they must be in the same cluster.
\end{lemma}
\begin{proof}
The proof follows by ``in the same cluster'' being a transitive property.

More formally, $u_1,u_2$ are in the same connected component in $(V,F)$, as $s(u_1)=s(u_2)$. Thus, there exists a path from $u_1$ to $u_2$. We claim that all nodes in a path must be in the same cluster. This is trivial if the path is of length $0$ ($u_1=u_2$) or of length $1$ ($u_1u_2\in F$). Else, the path is $u_1,w_1,\ldots,w_k,u_2$ for some $k\ge 1$. We inductively have that all of $w_1, \ldots, w_k, u_2$ must be in the same cluster, and $u_1$ must be in the same cluster with $w_1$ because $u_1w_1\in F$. Therefore, all nodes in the path must be in the same cluster with $w_1$.
\end{proof}

We now show that it is enough to bound the symmetric difference between $E$ and $E'$.

\begin{lemma} \label{lem:triangle}
The cost of our clustering \(C\) is $|E\symdiff E_C| \le (\alpha+1)|E\symdiff E'| + \alpha |E\symdiff {\OPT}|$.
\end{lemma}
\begin{proof}
The symmetric difference of sets satisfies the triangle inequality; we therefore have 
\begin{equation*}
  |E\symdiff E_C| \le |E\symdiff E'| + |E'\symdiff E_C|.
\end{equation*}
$C$ is an $\alpha$-approximation for $G'=(V, E')$ and thus
\(|E'\symdiff E_C| \le \alpha|E' \symdiff \OPT'| \le \alpha|E' \symdiff \OPT|\).
Therefore:
\begin{equation*}
  |E\symdiff E_C| \le |E\symdiff E'| + \alpha |E'\symdiff {\OPT}| \le |E\symdiff E'| + \alpha |E'\symdiff E| + \alpha |E\symdiff {\OPT}|.
\end{equation*}
with the second inequality following by applying the triangle inequality again.
\end{proof}

In order to upper bound $|E\symdiff E'|$ by the cost of the optimal clustering $|E\symdiff {\OPT}|$, we first need to lower bound the cost of the optimal clustering.

\begin{lemma} \label{lem:lowerBoundOPT}
Let $S$ be the set of all pairs of distinct supernodes $U,W$ that are in the same cluster in $\OPT$. Then $|E\symdiff {\OPT}| \ge \sum_{\{U,W\}\in S} |E(U,W)\symdiff E_3(U,W)|$.
\end{lemma}
\begin{proof}
The high-level idea is that when a node is connected to two hostile nodes, then any valid clustering needs to pay for at least one of these edges. Extending this fact to supernodes, we construct an edge set of size $\sum_{\{U,W\}\in S} |E(U,W)\symdiff E_3(U,W)|$ such that the optimal clustering needs to pay for each edge in this set.

First, for any $\{U,W\}\in S$ it holds that $E(U,W)\symdiff E_3(U,W) = E(U,W)\setminus E_3(U,W)$ because Line~\ref{line:connectSupernode} (\cref{alg:main}) does not modify edges between pairs of distinct supernodes, and Lines~\ref{line:disconnectSupernode} and~\ref{line:drop2} only remove edges.

Each edge of $E(U, W) \setminus E_3(U, W)$ is the result of applying Line~\ref{line:drop2},
seeing as Line~\ref{line:disconnectSupernode} only removes edges from hostile pairs of supernodes.
Thus each edge $u w \in E(U,W)\setminus E_3(U,W)$ can be paired up with a unique edge $x y \in E$ which is removed together with $u w$. Without loss of generality it holds that $x \in U, y \in Z$ for some supernode $Z$ different from $U$ and $W$.
Due to the way Line~\ref{line:drop2} chooses edges it must be the case that $Z$ and $W$ are hostile, hence $xy \in E \symdiff \OPT$.

Summing over all pairs of clustered supernodes gives the result stated in the lemma.
\end{proof}

We are now ready to bound $|E\symdiff E'|$. 

\begin{lemma} \label{lem:EEt}
$|E\symdiff E'| \le \myBeta |E\symdiff {\OPT}|$
\end{lemma}
\begin{proof}
To prove this, we first charge each pair of nodes in a way such that the total charge is at most $2|E\symdiff {\OPT}|$. Then we partition the pairs of nodes into $5$ different sets, and show that the size of the intersection between $E\symdiff E'$ and each of the $5$ sets is at most $\frac{1 + \sqrt5}{2}$ times the total charge given to the pairs in the given set.

The first three sets contain the pairs across non-hostile supernodes; out of them the first one is the most technically challenging, requiring a combination of Lemma~\ref{lem:lowerBoundOPT} (related to Line~\ref{line:drop2} of \cref{alg:main}) and a direct analysis on $E\symdiff \OPT$, as neither of them would suffice on their own. The analysis of the second and third sets relate to the rounding in Line~\ref{line:round}.
The fourth set contains pairs across hostile supernodes, while the fifth set contains pairs within supernodes. Their analysis is directly based on the hard constraints.

Let us define our charging scheme: first, each pair of nodes is charged if the optimal clustering pays for it, i.e.\ if this pair is in $E \symdiff \OPT$.
We further put a charge on the pairs \(uw \in E \symdiff E_3\) which connect supernodes that are clustered together in OPT.
Notice that the number of such edges is a lower bound on $|E\symdiff \OPT|$ by \cref{lem:lowerBoundOPT}.
Therefore the total charge over all pairs of nodes is at most~$2|E\symdiff \OPT|$ and no pair is charged twice.

\medskip
\emph{Case 1:}
Consider two distinct supernodes $U,W$ that are not hostile, which have more than~$\thrshld|U|\cdot |W|$ edges between them in $E$,
and have at most $\thrshld|U| \cdot |W|$ edges in~$E_3$. Then the rounding of Line~\ref{line:round} removes all edges between them.
Therefore $|E(U,W)\symdiff E'(U,W)| = |E(U,W)|\le |U|\cdot |W|$. If $\OPT$ separates $U$ and $W$, then the pairs are charged $|E(U,W)|$; else they are charged $|U|\cdot |W|-|E(U,W)|$ due to the part of the charging scheme related to~$E\symdiff \OPT$. In the latter case, they are also charged
$|E(U,W)|- |E_3(U,W)|$
due to the part of the charging scheme related to \cref{lem:lowerBoundOPT}. Therefore they are charged at least
\begin{align*}
  |U|\cdot |W|-|E(U,W)|+|E(U,W)|-|E_3(U,W)|
  &= |U| \cdot |W|-|E_3(U,W)| \\
  &\ge |U|\cdot |W|-\tfrac{3 - \sqrt5}{2}|U| \cdot |W|.
\end{align*}
Thus, in the worst case, these pairs contribute
\begin{equation*}
  \max\left\{\frac{|E(U,W)|}{|E(U,W)|}, \frac{|E(U,W)|}{|U|\cdot |W|-\thrshld{}|U|\cdot |W|}\right\}\le \frac1{1-\thrshld}=\frac{1 + \sqrt 5}{2}  
\end{equation*}
times more in $|E\symdiff E'|$ compared to their charge.

\medskip
\emph{Case 2:}
Consider two distinct supernodes $U,W$ that are not hostile, which have more than $\thrshld|U|\cdot |W|$ edges between them in $E$, and more than $\thrshld|U|\cdot |W|$ edges in~$E_3$. Then the rounding of Line~\ref{line:round} will include all $|U|\cdot |W|$ edges between them.
Thus we have $|E(U,W)\symdiff E'(U,W)| = |U|\cdot |W|-|E(U,W)| < (1-\thrshld)|U|\cdot |W|$.
If $\OPT$ separates~$U$ and $W$ it pays for $|E(U,W)|>\thrshld|U|\cdot |W|$ pairs.
Otherwise it pays $|U|\cdot |W|-|E(U,W)|$.
Thus, in the worst case, these pairs contribute
$\frac{1-\thrshld}{\thrshld}=\frac{1 + \sqrt 5}{2}$
times more in $|E\symdiff E'|$ compared to their charge.

\medskip
\emph{Case 3:}
If two distinct supernodes $U,W$ are not hostile and have at most~$\thrshld|U|\cdot |W|$ edges between them in $E$, then they also have at most that many edges in $E_3$ as we only remove edges between such supernodes. There are thus no edges between them in $E'$, meaning that $|E(U,W)\symdiff E'(U,W)| = |E(U,W)| \le \thrshld|U|\cdot |W|$. If~$\OPT$ separates $U,W$ it pays for $|E(U,W)|$ pairs related to the connection between~$U,W$; else it pays for $|U|\cdot |W|-|E(U,W)| \ge (1-\thrshld)|U|\cdot |W| > \thrshld |U|\cdot |W|$. Thus these pairs' contribution in $|E\symdiff E'|$ is at most as much as their charge.

\medskip
\emph{Case 4:}
Pairs $uv$ with $s(u)\ne s(v)$ and $s(u)$ hostile with $s(v)$ are not present in $E'$. That is because by Line~\ref{line:disconnectSupernode} no pair of hostile supernodes is connected; then Line~\ref{line:drop2} only removes edges, and Line~\ref{line:round} does not add any edge between $s(u)$ and $s(v)$ as they had $0 \le \thrshld|s(u)|\cdot |s(v)|$ edges between them. The edge $uv$ is also not present in $\OPT$ as $s(u)$ and $s(v)$ are not in the same cluster because they are hostile. These pairs' contribution in $|E\symdiff E'|$ is exactly equal to their charge.

\medskip
\emph{Case 5:}
Pairs $uv$ with $s(u)=s(v)$ are present in $E'$ by Line~\ref{line:connectSupernode} and the fact that all subsequent steps only modify edges whose endpoints are in different supernodes. The pair $uv$ is also present in $\OPT$, by \cref{lem:supernode}. Therefore these pairs' contribution in $|E\symdiff E'|$ is exactly equal to their charge.

\medskip
In the worst case, the pairs of each of the five sets contribute at most $\frac{1 + \sqrt5}{2}$ times more in $|E\symdiff E'|$ compared to their charge, which proves our lemma.
\end{proof}

We are now ready to prove the main theorem.

\begin{proof}[Proof of \cref{thm:CCCCombinatorial}]
In \cref{thm:CCCombinatorial} we established that there is a deterministic combinatorial \textsc{PIVOT} algorithm computing a Correlation Clustering with approximation factor $\alpha = 3 + \epsilon$ in time $\widetilde\Order(n^3)$, for any constant $\epsilon > 0$. Using this algorithm in \cref{alg:main} gives a valid clustering.
By \cref{lem:triangle,lem:EEt}, its approximation factor is bounded by $(\alpha + 1) \cdot (1 + \sqrt 5) + \alpha$. This is less than $16$ for $\epsilon = 0.01$.
\end{proof}
\section{PIVOT Algorithms for Correlation Clustering} \label{sec:deterministicPivoting}
\subsection{Lower Bound}
First we prove \cref{thm:pivot-lower-bound} which states that there is no \textsc{PIVOT} algorithm for Correlation Clustering with approximation factor better than $3$.

\begin{proof}[Proof of \cref{thm:pivot-lower-bound}]
Let $G=([2n],E)$ for some integer $n$, where the edge set $E$ contains all pairs of nodes except for pairs of the form $(2k+1,2k+2)$. In other words, the edge set of $G$ contains all edges except for a perfect matching.

Note that if we create a single cluster containing all nodes, then the cost is exactly $n$. On the other hand, let $u$ be the first choice that a \textsc{PIVOT} algorithm makes. If $u$ is even, let~$v=u-1$, otherwise let~$v=u+1$. By definition of $G$, $v$ is the only node not adjacent to~$u$. Therefore, the algorithm creates two clusters---one containing all nodes except for~$v$, and one containing only~$v$. There are $2n-2$ edges across the two clusters, and $n-1$ missing edges in the big cluster, meaning that the cost is $3n-3$.

Therefore, the approximation factor of any \textsc{PIVOT} algorithm is at least $(3n-3) / n = 3 - \frac3n$. This proves the theorem, as for any constant less than $3$, there exists a sufficiently large~$n$ such that $3-\frac3n$ is larger than that constant.
\end{proof}

\subsection{Optimal Deterministic PIVOT: 3-Approximation} \label{sec:pivoting-3}

A \emph{covering} LP is a linear program of the form \(\min_x \set{cx \mid Ax \geq b}\) where \(A, b, c\), and \(x\) are restricted to vectors and matrices of non-negative entries.
Covering LPs can be solved more efficiently than LPs in general and we rely on the following known machinery to prove \cref{thm:CCCombinatorial}:

\begin{theorem}[Covering LPs, Combinatorial~{{{\cite{GargK98,Fleischer04}}}}] \label{thm:covering-combinatorial}
Any covering LP with at most $N$ nonzero entries in the constraint matrix can be $(1 + \epsilon)$-approximated by a combinatorial algorithm in time $\widetilde\Order(N \epsilon^{-3})$.\footnote{The running time we state seems worse by a factor of $\epsilon^{-1}$ as compared to the theorems in~\cite{GargK98,Fleischer04}. This is because the authors assume access to a machine model with exact arithmetic of numbers of size exponential in $\epsilon^{-1}$. We can simulate this model using fixed-point arithmetic with a running time overhead of $\widetilde\Order(\epsilon^{-1})$.}
\end{theorem}

\begin{theorem}[Covering LPs, Non-Combinatorial~{{{\cite{Allen-ZhuO19,WangRM16}}}}] \label{thm:covering-noncombinatorial}
Any covering LP with at most $N$ nonzero entries in the constraint matrix can be $(1 + \epsilon)$-approximated in time $\widetilde\Order(N \epsilon^{-1})$.
\end{theorem}

Of the two theorems, the time complexity of the algorithm promised by \cref{thm:covering-noncombinatorial} is obviously better.
However, the algorithm of \cref{thm:covering-combinatorial} is remarkably simple in our setting and could thus prove to be faster in practice.
Note that either theorem suffices to obtain a \((3+\epsilon)\)-approximation for \CC{} in \(\widetilde{O}(n^3)\) time, for constant \(\epsilon>0\).

For completeness, and in order to demonstrate how simple the algorithm from \cref{thm:covering-combinatorial} is in our setting, we include the pseudocode as \cref{alg:charge}.
In \cref{app:charge} we formally prove that \cref{alg:charge} indeed has the properties promised by \cref{thm:covering-combinatorial}.

\begin{algorithm2e}[htb]
  \nonl\procedure{\Charge{$G=(V,E)$}}{
    Initialize $x_{uv},\, x^*_{uv} \gets 1$ for all $uv \in \binom V2$\;
    \While{$\sum_{uv} x_{uv} < B := (\binom n2 (1 + \epsilon))^{1/\epsilon} \,/\, (1 + \epsilon)$}{
      Find a bad triplet $uvw$ minimizing $x_{uv} + x_{vw} + x_{wu}$ \label{alg:charge:line:pick}\;
      $x_{uv} \gets (1 + \epsilon) \cdot x_{uv}$ \label{alg:charge:line:update-uv}\;
      $x_{vw} \gets (1 + \epsilon) \cdot x_{vw}$ \label{alg:charge:line:update-vw}\;
      $x_{wu} \gets (1 + \epsilon) \cdot x_{wu}$ \label{alg:charge:line:update-wu}\;
      \If{$(\sum_{uv} x_{uv}) \,/\, m(x) < (\sum_{uv} x^*_{uv}) \,/\, m(x^*)$}{   
        \lForEach{$uv\in \binom{V}{2}$}{$x^*_{uv} \gets x_{uv}$}
      }
    }
    \Return $\set{x^*_{uv} \,/\, m(x^*)}_{uv}$
  }

\caption{The combinatorial algorithm to $(1 + \Order(\epsilon))$-approximate the LP in \cref{fig:lps} (\cpageref{fig:lps}) using the multiplicative weights update method. The general method was given by Garg and Könemann~\cite{GargK98} and later refined by Fleischer~\cite{Fleischer04}. We here use the notation $m(x) = \min_{uvw \in T(G)} x_{uv} + x_{vw} + x_{wu}$.} \label{alg:charge}
\end{algorithm2e}

The solution found by \cref{alg:charge} is used together with the framework by van Zuylen and Williamson \cite{deterministicPivoting}, see \Cluster in \cref{alg:cluster}. \Cluster is discussed further in \cref{sec:pivoting-9}, where the following lemmas are proven.

\begin{algorithm2e}[!htb]
  \nonl\procedure{\Cluster{$G = (V, E), x = \set{x_{uv}}_{uv \in \binom V2}$}}{
    $C \gets \emptyset$\;
    \While{$V \neq \emptyset$}{
      Pick a pivot node $u \in V$ minimizing
      \begin{equation*}
        \frac{\displaystyle\sum_{vw : uvw \in T(G)} 1}{\displaystyle\sum_{vw : uvw \in T(G)} x_{vw}}
      \end{equation*} \label{alg:cluster:line:ratio}\;
      \vspace*{-3.5ex}
      Add a cluster containing $u$ and all its neighbors to $C$\;
      Remove $u$, its neighbors and all their incident edges from $G$\;
    }
    \Return $C$
  }

\caption{The \textsc{PIVOT} algorithm by van Zuylen and Williamson~\cite{deterministicPivoting}. Given a graph~$G$ and a good charging $\set{x_{uv}}_{uv}$ (in the sense of \cref{lem:cluster-correctness}), it computes a correlation clustering.} \label{alg:cluster}
\end{algorithm2e}

\begin{restatable}[Correctness of \textrm{\textsc{Cluster}}]{lem}{correctnessCluster}
\label{lem:cluster-correctness}
Assume that $x = \set{x_{uv}}_{uv}$ is a feasible solution to the LP in \cref{fig:lps}. Then $\Cluster{G, x}$ computes a correlation clustering of cost~\makebox{$3 \sum_{uv} x_{uv}$}. In particular, if $x$ is an $\alpha$-approximate solution to the LP (for some $\alpha \geq 1$), then $\Cluster{G, x}$ returns a $3\alpha$-approximate correlation clustering.
\end{restatable}

\begin{restatable}[Running Time of \textrm{\textsc{Cluster}}, {{{\cite{deterministicPivoting}}}}]{lem}{timeCluster}
\label{lem:cluster-time}
\Cluster{G, x} runs in time $\Order(n^3)$.
\end{restatable}

Given \cref{thm:covering-combinatorial,thm:covering-noncombinatorial} we quickly prove \cref{thm:CCCombinatorial}.

\begin{proof}[Proof of \cref{thm:CCCombinatorial}]
We compute a $(1 + \epsilon/3)$-approximate solution $x$ of the charging LP using \cref{thm:covering-combinatorial} (that is, using the procedure \Charge{G}). Plugging this solution $x$ into \Cluster{G, x} returns a $(3 + \epsilon)$-approximate correlation clustering by \cref{lem:cluster-correctness}. The total running time is bounded by $\Order(n^3)$ by \cref{lem:cluster-time} plus $\widetilde\Order(n^3 \epsilon^{-3})$ by \cref{thm:covering-combinatorial} (note that there are $n^3$ constraints, each affecting only a constant number of variables, hence the number of nonzeros in the constraint matrix is $N \leq \Order(n^3)$). For constant $\epsilon > 0$, this becomes $\widetilde\Order(n^3)$.

To obtain a $3$-approximation, we observe that any correlation clustering has cost less than~$\binom n2$. Hence, we can run the previous algorithm with $\epsilon = 1 / \binom n2$ and the $(3 + \epsilon)$-approximate solution is guaranteed to also be $3$-approximate. The running time would be bounded by $\widetilde\Order(n^9)$. To improve upon this, we use the covering LP solver in \cref{thm:covering-noncombinatorial} which runs in time $\widetilde\Order(n^3 \epsilon^{-1})$. By again setting~$\epsilon = 1 / \binom n2$, the running time becomes $\widetilde\Order(n^5)$.
\end{proof}

\clearpage
\bibliography{correlation}

\begin{thebibliography}{10}

\bibitem{autoLabelAgrawal}
Rakesh Agrawal, Alan Halverson, Krishnaram Kenthapadi, Nina Mishra, and Panayiotis Tsaparas.
\newblock Generating labels from clicks.
\newblock In Ricardo Baeza{-}Yates, Paolo Boldi, Berthier~A. Ribeiro{-}Neto, and Berkant~Barla Cambazoglu, editors, {\em Proceedings of the Second International Conference on Web Search and Web Data Mining, {WSDM} 2009, Barcelona, Spain, February 9-11, 2009}, pages 172--181. {ACM}, 2009.
\newblock \href {https://doi.org/10.1145/1498759.1498824} {\path{doi:10.1145/1498759.1498824}}.

\bibitem{charikarTree}
Nir Ailon and Moses Charikar.
\newblock Fitting tree metrics: Hierarchical clustering and phylogeny.
\newblock {\em {SIAM} J. Comput.}, 40(5):1275--1291, 2011.
\newblock Announced at FOCS'05.
\newblock \href {https://doi.org/10.1137/100806886} {\path{doi:10.1137/100806886}}.

\bibitem{pivoting}
Nir Ailon, Moses Charikar, and Alantha Newman.
\newblock Aggregating inconsistent information: Ranking and clustering.
\newblock {\em J. {ACM}}, 55(5):23:1--23:27, 2008.
\newblock Announced in STOC 2005.
\newblock \href {https://doi.org/10.1145/1411509.1411513} {\path{doi:10.1145/1411509.1411513}}.

\bibitem{Allen-ZhuO19}
Zeyuan Allen{-}Zhu and Lorenzo Orecchia.
\newblock Nearly linear-time packing and covering {LP} solvers -- achieving width-independence and -convergence.
\newblock {\em Math. Program.}, 175(1-2):307--353, 2019.
\newblock \href {https://doi.org/10.1007/s10107-018-1244-x} {\path{doi:10.1007/s10107-018-1244-x}}.

\bibitem{AlmanW21}
Josh Alman and Virginia~Vassilevska Williams.
\newblock A refined laser method and faster matrix multiplication.
\newblock In {\em Proceedings of the 2021 {ACM-SIAM} Symposium on Discrete Algorithms, {SODA} 2021, Virtual Conference, January 10 - 13, 2021}, pages 522--539. {SIAM}, 2021.
\newblock \href {https://doi.org/10.1137/1.9781611976465.32} {\path{doi:10.1137/1.9781611976465.32}}.

\bibitem{duplicate}
Arvind Arasu, Christopher R{\'{e}}, and Dan Suciu.
\newblock Large-scale deduplication with constraints using dedupalog.
\newblock In Yannis~E. Ioannidis, Dik~Lun Lee, and Raymond~T. Ng, editors, {\em Proceedings of the 25th International Conference on Data Engineering, {ICDE} 2009, March 29 2009 - April 2 2009, Shanghai, China}, pages 952--963. {IEEE} Computer Society, 2009.
\newblock \href {https://doi.org/10.1109/ICDE.2009.43} {\path{doi:10.1109/ICDE.2009.43}}.

\bibitem{sublinear}
Sepehr Assadi and Chen Wang.
\newblock Sublinear time and space algorithms for correlation clustering via sparse-dense decompositions.
\newblock {\em CoRR}, abs/2109.14528, 2021.
\newblock URL: \url{https://arxiv.org/abs/2109.14528}, \href {https://arxiv.org/abs/2109.14528} {\path{arXiv:2109.14528}}.

\bibitem{Bansal}
Nikhil Bansal, Avrim Blum, and Shuchi Chawla.
\newblock Correlation clustering.
\newblock {\em Mach. Learn.}, 56(1-3):89--113, 2004.
\newblock \href {https://doi.org/10.1023/B:MACH.0000033116.57574.95} {\path{doi:10.1023/B:MACH.0000033116.57574.95}}.

\bibitem{BehMany}
Soheil Behnezhad, Moses Charikar, Weiyun Ma, and Li{-}Yang Tan.
\newblock Almost 3-approximate correlation clustering in constant rounds.
\newblock In {\em 63rd {IEEE} Annual Symposium on Foundations of Computer Science, {FOCS} 2022, Denver, CO, USA, October 31 - November 3, 2022}, pages 720--731. {IEEE}, 2022.
\newblock \href {https://doi.org/10.1109/FOCS54457.2022.00074} {\path{doi:10.1109/FOCS54457.2022.00074}}.

\bibitem{BehSingle}
Soheil Behnezhad, Moses Charikar, Weiyun Ma, and Li{-}Yang Tan.
\newblock Single-pass streaming algorithms for correlation clustering.
\newblock In Nikhil Bansal and Viswanath Nagarajan, editors, {\em Proceedings of the 2023 {ACM-SIAM} Symposium on Discrete Algorithms, {SODA} 2023, Florence, Italy, January 22-25, 2023}, pages 819--849. {SIAM}, 2023.
\newblock URL: \url{https://doi.org/10.1137/1.9781611977554.ch33}, \href {https://doi.org/10.1137/1.9781611977554.CH33} {\path{doi:10.1137/1.9781611977554.CH33}}.

\bibitem{clusteringEnsembles}
Francesco Bonchi, Aristides Gionis, and Antti Ukkonen.
\newblock Overlapping correlation clustering.
\newblock {\em Knowl. Inf. Syst.}, 35(1):1--32, 2013.
\newblock \href {https://doi.org/10.1007/s10115-012-0522-9} {\path{doi:10.1007/s10115-012-0522-9}}.

\bibitem{differentialPrivacy}
Mark Bun, Marek Eli{\'{a}}s, and Janardhan Kulkarni.
\newblock Differentially private correlation clustering.
\newblock In Marina Meila and Tong Zhang, editors, {\em Proceedings of the 38th International Conference on Machine Learning, {ICML} 2021, 18-24 July 2021, Virtual Event}, volume 139 of {\em Proceedings of Machine Learning Research}, pages 1136--1146. {PMLR}, 2021.
\newblock URL: \url{http://proceedings.mlr.press/v139/bun21a.html}.

\bibitem{dynamicStream}
Melanie Cambus, Fabian Kuhn, Etna Lindy, Shreyas Pai, and Jara Uitto.
\newblock {\em A ($3 + \varepsilon$)-Approximate Correlation Clustering Algorithm in Dynamic Streams}, pages 2861--2880.
\newblock {SIAM}, 2024.
\newblock URL: \url{https://epubs.siam.org/doi/abs/10.1137/1.9781611977912.101}, \href {https://arxiv.org/abs/https://epubs.siam.org/doi/pdf/10.1137/1.9781611977912.101} {\path{arXiv:https://epubs.siam.org/doi/pdf/10.1137/1.9781611977912.101}}, \href {https://doi.org/10.1137/1.9781611977912.101} {\path{doi:10.1137/1.9781611977912.101}}.

\bibitem{clusterLP}
Nairen Cao, Vincent Cohen{-}Addad, Euiwoong Lee, Shi Li, Alantha Newman, and Lukas Vogl.
\newblock Understanding the cluster linear program for correlation clustering.
\newblock In Bojan Mohar, Igor Shinkar, and Ryan O'Donnell, editors, {\em Proceedings of the 56th Annual {ACM} Symposium on Theory of Computing, {STOC} 2024, Vancouver, BC, Canada, June 24-28, 2024}, pages 1605--1616. {ACM}, 2024.
\newblock \href {https://doi.org/10.1145/3618260.3649749} {\path{doi:10.1145/3618260.3649749}}.

\bibitem{sub3parallel}
Nairen Cao, Shang-En Huang, and Hsin-Hao SU.
\newblock {\em Breaking 3-Factor Approximation for Correlation Clustering in Polylogarithmic Rounds}, pages 4124--4154.
\newblock {SIAM}, 2024.
\newblock URL: \url{https://epubs.siam.org/doi/abs/10.1137/1.9781611977912.143}, \href {https://arxiv.org/abs/https://epubs.siam.org/doi/pdf/10.1137/1.9781611977912.143} {\path{arXiv:https://epubs.siam.org/doi/pdf/10.1137/1.9781611977912.143}}, \href {https://doi.org/10.1137/1.9781611977912.143} {\path{doi:10.1137/1.9781611977912.143}}.

\bibitem{autoLabelChakrabarti}
Deepayan Chakrabarti, Ravi Kumar, and Kunal Punera.
\newblock A graph-theoretic approach to webpage segmentation.
\newblock In Jinpeng Huai, Robin Chen, Hsiao{-}Wuen Hon, Yunhao Liu, Wei{-}Ying Ma, Andrew Tomkins, and Xiaodong Zhang, editors, {\em Proceedings of the 17th International Conference on World Wide Web, {WWW} 2008, Beijing, China, April 21-25, 2008}, pages 377--386. {ACM}, 2008.
\newblock \href {https://doi.org/10.1145/1367497.1367549} {\path{doi:10.1145/1367497.1367549}}.

\bibitem{MakarySingle}
Sayak Chakrabarty and Konstantin Makarychev.
\newblock Single-pass pivot algorithm for correlation clustering. keep it simple!
\newblock {\em CoRR}, abs/2305.13560, 2023.
\newblock URL: \url{https://doi.org/10.48550/arXiv.2305.13560}, \href {https://arxiv.org/abs/2305.13560} {\path{arXiv:2305.13560}}, \href {https://doi.org/10.48550/ARXIV.2305.13560} {\path{doi:10.48550/ARXIV.2305.13560}}.

\bibitem{cutRadius}
Moses Charikar, Venkatesan Guruswami, and Anthony Wirth.
\newblock Clustering with qualitative information.
\newblock {\em J. Comput. Syst. Sci.}, 71(3):360--383, 2005.
\newblock Announced in FOCS 2003.
\newblock \href {https://doi.org/10.1016/j.jcss.2004.10.012} {\path{doi:10.1016/j.jcss.2004.10.012}}.

\bibitem{NearOptimal2}
Shuchi Chawla, Konstantin Makarychev, Tselil Schramm, and Grigory Yaroslavtsev.
\newblock Near optimal {LP} rounding algorithm for correlation clustering on complete and complete k-partite graphs.
\newblock In Rocco~A. Servedio and Ronitt Rubinfeld, editors, {\em Proceedings of the Forty-Seventh Annual {ACM} on Symposium on Theory of Computing, {STOC} 2015, Portland, OR, USA, June 14-17, 2015}, pages 219--228. {ACM}, 2015.
\newblock \href {https://doi.org/10.1145/2746539.2746604} {\path{doi:10.1145/2746539.2746604}}.

\bibitem{commDet}
Yudong Chen, Sujay Sanghavi, and Huan Xu.
\newblock Clustering sparse graphs.
\newblock In Peter~L. Bartlett, Fernando C.~N. Pereira, Christopher J.~C. Burges, L{\'{e}}on Bottou, and Kilian~Q. Weinberger, editors, {\em Advances in Neural Information Processing Systems 25: 26th Annual Conference on Neural Information Processing Systems 2012. Proceedings of a meeting held December 3-6, 2012, Lake Tahoe, Nevada, United States}, pages 2213--2221, 2012.
\newblock URL: \url{https://proceedings.neurips.cc/paper/2012/hash/1e6e0a04d20f50967c64dac2d639a577-Abstract.html}.

\bibitem{CohenLS19}
Michael~B. Cohen, Yin~Tat Lee, and Zhao Song.
\newblock Solving linear programs in the current matrix multiplication time.
\newblock In {\em Proceedings of the 51st Annual {ACM} {SIGACT} Symposium on Theory of Computing, {STOC} 2019, Phoenix, AZ, USA, June 23-26, 2019}, pages 938--942. {ACM}, 2019.
\newblock \href {https://doi.org/10.1145/3313276.3316303} {\path{doi:10.1145/3313276.3316303}}.

\bibitem{l0}
Vincent Cohen{-}Addad, Chenglin Fan, Euiwoong Lee, and Arnaud de~Mesmay.
\newblock Fitting metrics and ultrametrics with minimum disagreements.
\newblock In {\em 63rd {IEEE} Annual Symposium on Foundations of Computer Science, {FOCS} 2022, Denver, CO, USA, October 31 - November 3, 2022}, pages 301--311. {IEEE}, 2022.
\newblock \href {https://doi.org/10.1109/FOCS54457.2022.00035} {\path{doi:10.1109/FOCS54457.2022.00035}}.

\bibitem{parallel}
Vincent Cohen{-}Addad, Silvio Lattanzi, Slobodan Mitrovic, Ashkan Norouzi{-}Fard, Nikos Parotsidis, and Jakub Tarnawski.
\newblock Correlation clustering in constant many parallel rounds.
\newblock In Marina Meila and Tong Zhang, editors, {\em Proceedings of the 38th International Conference on Machine Learning, {ICML} 2021, 18-24 July 2021, Virtual Event}, volume 139 of {\em Proceedings of Machine Learning Research}, pages 2069--2078. {PMLR}, 2021.
\newblock URL: \url{http://proceedings.mlr.press/v139/cohen-addad21b.html}.

\bibitem{apx173}
Vincent Cohen{-}Addad, Euiwoong Lee, Shi Li, and Alantha Newman.
\newblock Handling correlated rounding error via preclustering: {A} 1.73-approximation for correlation clustering.
\newblock In {\em 64th {IEEE} Annual Symposium on Foundations of Computer Science, {FOCS} 2023, Santa Cruz, CA, USA, November 6-9, 2023}, pages 1082--1104. {IEEE}, 2023.
\newblock \href {https://doi.org/10.1109/FOCS57990.2023.00065} {\path{doi:10.1109/FOCS57990.2023.00065}}.

\bibitem{sub2}
Vincent Cohen{-}Addad, Euiwoong Lee, and Alantha Newman.
\newblock Correlation clustering with sherali-adams.
\newblock {\em CoRR}, abs/2207.10889, 2022.
\newblock \href {https://arxiv.org/abs/2207.10889} {\path{arXiv:2207.10889}}, \href {https://doi.org/10.48550/arXiv.2207.10889} {\path{doi:10.48550/arXiv.2207.10889}}.

\bibitem{combinatorialMikkel}
Vincent Cohen{-}Addad, David~Rasmussen Lolck, Marcin Pilipczuk, Mikkel Thorup, Shuyi Yan, and Hanwen Zhang.
\newblock Combinatorial correlation clustering.
\newblock In Bojan Mohar, Igor Shinkar, and Ryan O'Donnell, editors, {\em Proceedings of the 56th Annual {ACM} Symposium on Theory of Computing, {STOC} 2024, Vancouver, BC, Canada, June 24-28, 2024}, pages 1617--1628. {ACM}, 2024.
\newblock \href {https://doi.org/10.1145/3618260.3649712} {\path{doi:10.1145/3618260.3649712}}.

\bibitem{weighted}
Erik~D. Demaine, Dotan Emanuel, Amos Fiat, and Nicole Immorlica.
\newblock Correlation clustering in general weighted graphs.
\newblock {\em Theor. Comput. Sci.}, 361(2-3):172--187, 2006.
\newblock \href {https://doi.org/10.1016/j.tcs.2006.05.008} {\path{doi:10.1016/j.tcs.2006.05.008}}.

\bibitem{Fleischer04}
Lisa Fleischer.
\newblock A fast approximation scheme for fractional covering problems with variable upper bounds.
\newblock In {\em Proceedings of the Fifteenth Annual {ACM-SIAM} Symposium on Discrete Algorithms, {SODA} 2004, New Orleans, Louisiana, USA, January 11-14, 2004}, pages 1001--1010. {SIAM}, 2004.
\newblock URL: \url{http://dl.acm.org/citation.cfm?id=982792.982942}.

\bibitem{fpt}
Fedor~V. Fomin, Stefan Kratsch, Marcin Pilipczuk, Michal Pilipczuk, and Yngve Villanger.
\newblock Tight bounds for parameterized complexity of cluster editing with a small number of clusters.
\newblock {\em J. Comput. Syst. Sci.}, 80(7):1430--1447, 2014.
\newblock \href {https://doi.org/10.1016/j.jcss.2014.04.015} {\path{doi:10.1016/j.jcss.2014.04.015}}.

\bibitem{GargK98}
Naveen Garg and Jochen Koenemann.
\newblock Faster and simpler algorithms for multicommodity flow and other fractional packing problems.
\newblock In {\em Proceedings of the 39th Annual Symposium on Foundations of Computer Science}, FOCS '98, page 300, USA, 1998. IEEE Computer Society.

\bibitem{Janson18}
Svante Janson.
\newblock Tail bounds for sums of geometric and exponential variables.
\newblock {\em Statistics \& Probability Letters}, 135(C):1--6, 2018.
\newblock \href {https://doi.org/10.1016/j.spl.2017.11.017} {\path{doi:10.1016/j.spl.2017.11.017}}.

\bibitem{disambiguation}
Dmitri~V. Kalashnikov, Zhaoqi Chen, Sharad Mehrotra, and Rabia Nuray{-}Turan.
\newblock Web people search via connection analysis.
\newblock {\em {IEEE} Trans. Knowl. Data Eng.}, 20(11):1550--1565, 2008.
\newblock \href {https://doi.org/10.1109/TKDE.2008.78} {\path{doi:10.1109/TKDE.2008.78}}.

\bibitem{image1}
Sungwoong Kim, Sebastian Nowozin, Pushmeet Kohli, and Chang~Dong Yoo.
\newblock Higher-order correlation clustering for image segmentation.
\newblock In John Shawe{-}Taylor, Richard~S. Zemel, Peter~L. Bartlett, Fernando C.~N. Pereira, and Kilian~Q. Weinberger, editors, {\em Advances in Neural Information Processing Systems 24: 25th Annual Conference on Neural Information Processing Systems 2011. Proceedings of a meeting held 12-14 December 2011, Granada, Spain}, pages 1530--1538, 2011.
\newblock URL: \url{https://proceedings.neurips.cc/paper/2011/hash/98d6f58ab0dafbb86b083a001561bb34-Abstract.html}.

\bibitem{global}
Domenico Mandaglio, Andrea Tagarelli, and Francesco Gullo.
\newblock Correlation clustering with global weight bounds.
\newblock In Nuria Oliver, Fernando P{\'{e}}rez{-}Cruz, Stefan Kramer, Jesse Read, and Jos{\'{e}}~Antonio Lozano, editors, {\em Machine Learning and Knowledge Discovery in Databases. Research Track - European Conference, {ECML} {PKDD} 2021, Bilbao, Spain, September 13-17, 2021, Proceedings, Part {II}}, volume 12976 of {\em Lecture Notes in Computer Science}, pages 499--515. Springer, 2021.
\newblock \href {https://doi.org/10.1007/978-3-030-86520-7\_31} {\path{doi:10.1007/978-3-030-86520-7\_31}}.

\bibitem{local}
Gregory~J. Puleo and Olgica Milenkovic.
\newblock Correlation clustering with constrained cluster sizes and extended weights bounds.
\newblock {\em {SIAM} J. Optim.}, 25(3):1857--1872, 2015.
\newblock \href {https://doi.org/10.1137/140994198} {\path{doi:10.1137/140994198}}.

\bibitem{Brand20}
Jan van~den Brand.
\newblock A deterministic linear program solver in current matrix multiplication time.
\newblock In {\em Proceedings of the 2020 {ACM-SIAM} Symposium on Discrete Algorithms, {SODA} 2020, Salt Lake City, UT, USA, January 5-8, 2020}, pages 259--278. {SIAM}, 2020.
\newblock \href {https://doi.org/10.1137/1.9781611975994.16} {\path{doi:10.1137/1.9781611975994.16}}.

\bibitem{deterministicPivoting}
Anke van Zuylen and David~P. Williamson.
\newblock Deterministic pivoting algorithms for constrained ranking and clustering problems.
\newblock {\em Math. Oper. Res.}, 34(3):594--620, 2009.
\newblock Announced in SODA 2007.
\newblock \href {https://doi.org/10.1287/moor.1090.0385} {\path{doi:10.1287/moor.1090.0385}}.

\bibitem{deterministic}
Nate Veldt.
\newblock Correlation clustering via strong triadic closure labeling: Fast approximation algorithms and practical lower bounds.
\newblock In Kamalika Chaudhuri, Stefanie Jegelka, Le~Song, Csaba Szepesv{\'{a}}ri, Gang Niu, and Sivan Sabato, editors, {\em International Conference on Machine Learning, {ICML} 2022, 17-23 July 2022, Baltimore, Maryland, {USA}}, volume 162 of {\em Proceedings of Machine Learning Research}, pages 22060--22083. {PMLR}, 2022.
\newblock URL: \url{https://proceedings.mlr.press/v162/veldt22a.html}.

\bibitem{communityVeldt}
Nate Veldt, David~F. Gleich, and Anthony Wirth.
\newblock A correlation clustering framework for community detection.
\newblock In Pierre{-}Antoine Champin, Fabien Gandon, Mounia Lalmas, and Panagiotis~G. Ipeirotis, editors, {\em Proceedings of the 2018 World Wide Web Conference on World Wide Web, {WWW} 2018, Lyon, France, April 23-27, 2018}, pages 439--448. {ACM}, 2018.
\newblock \href {https://doi.org/10.1145/3178876.3186110} {\path{doi:10.1145/3178876.3186110}}.

\bibitem{alias}
Michael~D. Vose.
\newblock A linear algorithm for generating random numbers with a given distribution.
\newblock {\em {IEEE} Trans. Software Eng.}, 17(9):972--975, 1991.
\newblock \href {https://doi.org/10.1109/32.92917} {\path{doi:10.1109/32.92917}}.

\bibitem{Walker74}
Alastair~J. Walker.
\newblock New fast method for generating discrete random numbers with arbitrary frequency distributions.
\newblock {\em Electronics Letters}, 10:127--128(1), April 1974.

\bibitem{WangRM16}
Di~Wang, Satish Rao, and Michael~W. Mahoney.
\newblock Unified acceleration method for packing and covering problems via diameter reduction.
\newblock In {\em 43rd International Colloquium on Automata, Languages, and Programming, {ICALP} 2016, July 11-15, 2016, Rome, Italy}, volume~55 of {\em LIPIcs}, pages 50:1--50:13. Schloss Dagstuhl - Leibniz-Zentrum f{\"{u}}r Informatik, 2016.
\newblock \href {https://doi.org/10.4230/LIPIcs.ICALP.2016.50} {\path{doi:10.4230/LIPIcs.ICALP.2016.50}}.

\bibitem{image2}
Julian Yarkony, Alexander~T. Ihler, and Charless~C. Fowlkes.
\newblock Fast planar correlation clustering for image segmentation.
\newblock In Andrew~W. Fitzgibbon, Svetlana Lazebnik, Pietro Perona, Yoichi Sato, and Cordelia Schmid, editors, {\em Computer Vision - {ECCV} 2012 - 12th European Conference on Computer Vision, Florence, Italy, October 7-13, 2012, Proceedings, Part {VI}}, volume 7577 of {\em Lecture Notes in Computer Science}, pages 568--581. Springer, 2012.
\newblock \href {https://doi.org/10.1007/978-3-642-33783-3\_41} {\path{doi:10.1007/978-3-642-33783-3\_41}}.

\end{thebibliography}
\clearpage
\appendix
\section{Analysis of CHARGE}
\label{app:charge}
In this appendix we prove that \Charge (\cref{alg:charge}, \cpageref{alg:charge}) computes a \((1+\Order(\epsilon))\)-approximation to the charging LP of \cref{fig:lps} in time \(\Order(n^3 \poly(\log n / \epsilon))\), thus matching the guarantees of \cref{thm:covering-combinatorial}.
This algorithm was first developed by Garg and Könemann~\cite{GargK98} and later refined by Fleischer~\cite{Fleischer04}.

We analyze the algorithm using a primal--dual approach: We first argue that \Charge{G} constructs a good solution (up to rescaling) to the primal LP, and then compare this to an optimal dual solution.
The gap between primal and dual value is bounded by $1 + \Order(\epsilon)$, and by the weak duality theorem it follows that the primal solution computed by the algorithm is $(1 + \Order(\epsilon))$-approximately optimal.

Let us introduce some notation: Let $t$ denote the total number of iterations of \Charge{G}. For an iteration $k$, let $x_{uv}^{(k)}$ denote the current value of $x_{uv}$. We also write~$s^{(k)} = \sum_{uv} x_{uv}^{(k)}$ and $m^{(k)} = m(x^{(k)}) = \min_{uvw \in T(G)} x^{(k)}_{uv} + x^{(k)}_{vw} + x^{(k)}_{wu}$. Finally, we set $s = \sum_{uv} x^*_{uv}$ and~$m = m(x^*)$. We have $s / m \leq s^{(k)} / m^{(k)}$ for all iterations $k$ by the way that~$x^*$ is constructed.

\begin{observation}[Primal Solution] \label{obs:primal}
$\set{x^*_{uv} / m}_{uv}$ is a feasible primal solution with value~$s / m$.
\end{observation}
\begin{proof}
For any bad triplet $uvw$ we have that $x^*_{uv} + x^*_{vw} + x^*_{wu} \geq m$ by definition. Hence $\set{x_{uv}^* / m}$ satisfies all primal constraints and is feasible. Its value is~$s / m$ by definition.
\end{proof}

\begin{observation}[Dual Solution] \label{obs:dual}
Let $y_{uvw}$ be the number of iterations in which $uvw$ was picked in Line~\ref{alg:charge:line:pick}, scaled by $(\log_{1+\epsilon}(B) + 1)^{-1}$. Then $\set{y_{uvw}}_{uvw}$ is a feasible dual solution with value $t / (\log_{1+\epsilon}(B) + 1)$.
\end{observation}
\begin{proof}
In order to prove feasibility, we need to argue that $uvw$ is selected in at most $\log_{1+\epsilon}(B) + 1$ iterations. Indeed, in every iteration where $uvw$ is picked we multiplicatively increase $x_{uv}$ by $1 + \epsilon$. This can happen at most $\log_{1+\epsilon}(B) + 1$ times before the loop terminates. The value of $\set{y_{uvw}}_{uvw}$ is $\sum_{uvw} y_{uvw} = t / (\log_{1+\epsilon}(B) + 1)$.
\end{proof}

We additionally need the following technical lemma.

\begin{lemma} \label{lem:primal-increment}
For any iteration $k$, it holds that $s^{(k)} \leq \binom n2 \cdot \exp(\epsilon k m / s)$.
\end{lemma}
\begin{proof}
The proof is by induction. For $k = 0$ the statement is clear since we initially assign~$x_{uv} \gets 1$ for all pairs $uv$. For $k > 0$ we have
\begin{align}
  s^{(k)}
  &= s^{(k-1)} + \epsilon m^{(k-1)} \label{lem:primal-increment:eq:1} \\
  &\leq s^{(k-1)} \cdot (1 + \epsilon m / s) \label{lem:primal-increment:eq:2} \\
  &\leq \textstyle\binom n2 \cdot \exp(\epsilon (k-1) m / s) \cdot (1 + \epsilon m / s) \label{lem:primal-increment:eq:3} \\
  &\leq \textstyle\binom n2 \cdot \exp(\epsilon k m / s), \label{lem:primal-increment:eq:4}
\end{align}
where we used~\eqref{lem:primal-increment:eq:1} the update rule in Lines~\ref{alg:charge:line:update-uv}--\ref{alg:charge:line:update-wu},~\eqref{lem:primal-increment:eq:2} the fact that $s/m \leq s^{(k)} / m^{(k)}$,~\eqref{lem:primal-increment:eq:3} the induction hypothesis and~\eqref{lem:primal-increment:eq:4} the fact that $1 + x \leq \exp(x)$ for all real $x$.
\end{proof}

In combination we obtain the correctness of \Charge{G}:

\begin{lemma}[Correctness of \textrm{\textsc{Charge}}] \label{lem:charge-correctness}
The algorithm \Charge{G} correctly computes a~$(1 + \Order(\epsilon))$-approximate solution $\set{x^*_{uv} / m}_{uv}$ to the charging LP.
\end{lemma}
\begin{proof}
In order to argue that the primal solution $\set{x^*_{uv} / m}_{uv}$ from \cref{obs:primal} is $(1+\Order(\epsilon))$-approximate, it suffices to bound the gap to its corresponding dual solution from \cref{obs:dual}. Their gap is bounded by
\begin{align*}
  \frac{s / m}{t / (\log_{1+\epsilon}(B) + 1)}
  &= \frac{s (\log_{1+\epsilon}(B) + 1)}{mt}
\intertext{Recall that the algorithm terminates with $s \geq B$, thus by \cref{lem:primal-increment} we obtain that $B \leq \binom n2 \cdot \exp(\epsilon tm / s)$, or equivalently $tm / s \geq \epsilon^{-1} \ln(B / \binom n2)$. It follows that the gap is bounded by}
  &\leq \frac{\epsilon(\log_{1+\epsilon}(B) + 1)}{\ln(B / \binom n2)} \\
  &= \frac{\ln(B (1 + \epsilon))}{\ln(B / \binom n2)} \cdot \frac{\epsilon}{\ln(1 + \epsilon)}
\intertext{By setting $B = (\binom n2 (1 + \epsilon))^{1/\epsilon} / (1 + \epsilon)$ as in the algorithm, the first term becomes $1 / (1 - \epsilon)$ and thus}
  &= \frac{1}{1 - \epsilon} \cdot \frac{\epsilon}{\ln(1 + \epsilon)} \\
  &\leq \frac{1}{1 - \epsilon} \cdot \frac{\epsilon}{\epsilon - \epsilon^2 / 2} \\
  &\leq 1 + \Order(\epsilon). \qedhere
\end{align*}
\end{proof}

\begin{lemma}[Running Time of \textrm{\textsc{Charge}}] \label{lem:charge-time}
The running time of \Charge{G} is bounded by $\Order(n^3 \poly(\log n / \epsilon))$.
\end{lemma}
\begin{proof}
Any variable $x_{uv}$ can be increased at most $\log_{1+\epsilon}(B) + 1$ times. Hence, the total number of iterations is bounded by $\Order(n^2 \log_{1+\epsilon}(B)) = \Order(n^2 \poly(\log n / \epsilon))$. To efficiently implement the loop, we use a priority queue to maintain $\set{x_{uv} + x_{vw} + x_{wu}}_{uvw \in T(G)}$. The initialization takes time $\Order(n^3 \log n)$. In each iteration we can select the minimum-weight bad triplet in time $\Order(\log n)$ by a single query. Changing the three variables $x_{uv}, x_{vw}, x_{wu}$ affects at most $\Order(n)$ entries in the queue and therefore takes time~$\Order(n \log n)$.

In the previous paragraph we assumed that arithmetic operations run in unit time. However, observe that we work with numbers of magnitude up to $B$ and precision $\epsilon$. We can perform arithmetic operations on numbers of that size in time $\poly(\log(B / \epsilon)) = \poly(\log n / \epsilon)$. Therefore, the total running time increases by a factor $\poly(\log n / \epsilon)$ and is still bounded by $\Order(n^3 \poly(\log n / \epsilon))$ as claimed.
\end{proof}

\section{Correctness of CLUSTER} \label{sec:pivoting-9}
We borrow the algorithm from van Zuylen and Williamson~\cite{deterministicPivoting}, see \Cluster in \cref{alg:cluster}. We present our analysis using the charging LP relaxation in \cref{lem:cluster-correctness}. To obtain a best-possible \textsc{PIVOT} algorithm, think of the input $x$ as an exact solution to the LP in \cref{fig:lps}. We denote its optimal value by $\OPT^{(\text{LP})}$ and the optimal value of the correlation clustering by $\OPT^{(\text{CC})}$.

\correctnessCluster*
\begin{proof}
Let $G^{(i)} = (V^{(i)}, E^{(i)})$ denote the graph $G$ after the $i$-th iteration of the loop, i.e.,~$G^{(0)}$ is the initial graph $G$ and $G^{(t)}$ is the empty graph for $t$ the total number of iterations. Let~$u_i$ denote the pivot node selected in the $i$-th iteration of the algorithm. It is easy to check that the total number of violated preferences in the clustering $C$ is equal to
\begin{equation*}
  \sum_{i=0}^{t-1} \sum_{\substack{vw :\\u_ivw \in T(G^{(i)})}} 1
\end{equation*}
Indeed, in the $i$-th iteration we violate exactly the negative preferences of pairs $vw$ which are both neighbors of $u$, and the positive preferences of pairs $vw$ for which exactly one is a neighbor of $u$. In any such case and only in these cases, $uvw$ is a bad triplet. Using that $x$ is a feasible LP solution, we obtain the following bound:  
\begin{align*}
  \sum_{u \in V^{(i)}} \sum_{\substack{vw :\\uvw \in T(G^{(i)})}} 1
  &= 3 \cdot \sum_{uvw \in T(G^{(i)})} 1 \\
  &\leq 3 \cdot \sum_{uvw \in T(G^{(i)})} x_{uv} + x_{vw} + x_{wu} \\
  &= 3 \cdot \sum_{u \in V^{(i)}} \sum_{\substack{vw :\\uvw \in T(G^{(i)})}} x_{vw}
\end{align*}
By the way we picked $u_i$ in Line~\ref{alg:cluster:line:ratio}, we have that $\sum_{\substack{vw : uvw \in T(G^{(i)})}} 1 \leq 3 \cdot \sum_{\substack{vw : uvw \in T(G^{(i)})}} x_{vw}$. It follows that the total cost of the clustering is bounded by
\begin{align*}
  \sum_{i=0}^{t-1} \sum_{\substack{vw :\\u_ivw \in T(G^{(i)})}} 1
  &\leq 3 \cdot \sum_{i=0}^{t-1} \sum_{\substack{vw :\\u_ivw \in T(G^{(i)})}} x_{vw} \\
  &\leq 3 \cdot \sum_{vw \in \binom V2} x_{vw}
\end{align*}
Here, we used that every pair of vertices is counted for in exactly one graph $G^{(i)}$. This finishes the first part of the lemma.

For the second part, assume that $x$ is an $\alpha$-approximate optimal solution, i.e., $\sum_{uv} x_{uv} \leq \alpha \OPT^{(\text{LP})}$. We claim that $\OPT^{(\text{LP})} \leq \OPT^{(\text{CC})}$. Indeed, we can plug in any correlation clustering into the primal LP as follows: For every pair $uv$ whose preference is violated set $x_{uv} = 1$ and for all other pairs set $x_{uv} = 0$. The important observation is that in any correlation clustering solution, we charge at least one edge in every bad triplet. Hence, the constraints of the LP are satisfied, and we obtain a feasible solution of value $\OPT^{(\text{CC})}$. It follows that the correlation clustering constructed by \Cluster{G, x} has cost at most
\begin{equation*}
  3 \cdot \sum_{uv \in \binom V2} x_{uv} \leq 3\alpha \cdot \OPT^{(\text{LP})} \leq 3\alpha \cdot \OPT^{(\text{CC})} \qedhere
\end{equation*}
\end{proof}

\timeCluster*
\begin{proof}
We can efficiently implement \Cluster{G, x} by first precomputing $\sum_{vw : uvw \in T(G)} 1$ and $\sum_{vw : uvw \in T(G)} x_{vw}$ for every node $u$ in time $\Order(n^3)$. Then in the remaining algorithm we can efficiently select the pivot (for instance, by exhaustively checking all nodes $u$) and remove its cluster from the graph. For every vertex $u$ which is removed from the graph in that way, we can enumerate all bad triplets involving $u$ and update the precomputed quantities appropriately. Since every node is removed exactly once, the total running time is bounded by $\Order(n^3)$.  
\end{proof}
\section{Analysis of \CCC{}} \label{sec:analysisCCC}

\subparagraph*{Running Time.} We first prove the running time of our algorithm.

\begin{lemma} \label{lem:runningTime}
The running time of Algorithm~\ref{alg:main} is $O(n(n+m) + T(n,\binom{n}{2}))$, where $T(n',m')$ is an upper bound on the running time of the PIVOT algorithm we use on a graph with $n'$ nodes and $m'$ edges.
\end{lemma}
\begin{proof}
Computing the connected components of $(V,F)$ takes $O(n+|F|)$ time. Adding all the edges between supernodes takes $O(n^2)$ time. Then we can contract the supernodes (allowing parallel edges) in $O(n^2)$ time. Removing the edges between hostile supernodes takes $O(m+|H|)$ time. 

For the steps in the loop of Line~\ref{line:drop2}, notice that there are at most $m$ edges connecting distinct supernodes, as the only edges we added were internal in supernodes. We can iterate over all these edges $uv$, and over all supernodes $W$. If $W$ is connected with $s(u)$ and hostile with $s(v)$, then we remove $uv$ and an arbitrary edge connecting $W$ with $s(u)$, and similarly if $W$ is connected with $s(v)$ and hostile with $s(u)$. This takes $O(n \cdot m)$ time. Each pair of edges removed trivially satisfies the requirements of Line~\ref{line:drop2}. As we do not add edges in this step, it is impossible that when finishing there is still a pair of edges $e_1,e_2$ that needed to be removed; when processing $e_1$, we would remove $e_1$ along with some other edge (possibly different from $e_2$).

Rounding the connections between pairs of supernodes is done in $O(n^2)$ time.

The final graph may have at most $\binom{n}{2}$ edges (even if $m$ was much smaller, e.g.\ in the case where all nodes belong in the same supernode), therefore the time spent by the \textsc{PIVOT} algorithm is at most $T(n,\binom{n}{2})$.

The claimed bound follows by both $|F|$ and $|H|$ being $O(n^2)$.
\end{proof}

\subparagraph*{Correctness.}
We finally prove  correctness---that is, we prove that either the final algorithm satisfies all hard constraints, or no clustering can satisfy the hard constraints and the algorithm outputs ``Impossible''.

We start with showing that our algorithm correctly detects all cases where the hard constraints are impossible to satisfy.

\begin{lemma} \label{lem:impossible}
Given an instance $(V,E,F,H)$ of \CCHC, the graph~$G'\gets\Transform(V,E,F,H)$ is equal to $(\emptyset,\emptyset)$ if and only if the \CCHC{} instance is impossible to satisfy.
\end{lemma}
\begin{proof}
We show that if two hostile nodes are in the same supernode, then the instance is not satisfiable and the algorithm correctly determines it; on the other hand, if no such hostile nodes exist, then there exists at least one valid clustering.

It holds that $G'=(\emptyset,\emptyset)$ if there exist nodes $u_1,u_2$ such that $u_1,u_2$ are in the same connected component of $(V,F)$ and $u_1u_2\in H$. Then $u_1,u_2$ must be in the same cluster (Lemma~\ref{lem:supernode}) and not be in the same cluster (because $u_1u_2\in H$). Therefore the instance is impossible to satisfy.

Otherwise, no $u_1,u_2$ in the same supernode are hostile. Creating a cluster for each supernode is a valid clustering. To see this, notice that no hostility constraint is violated, by hypothesis. All friendliness constraints are satisfied because any two nodes that must be linked belong in the same supernode, and thus in the same cluster. Therefore such an instance is satisfiable.
\end{proof}

In the following we can thus assume that we have a satisfiable instance with no supernode being hostile to itself.  The next lemma shows that friendly nodes have the same neighborhood and are connected.

\begin{lemma} \label{lem:friendlyStructure}
Given a satisfiable instance $(V,E,F,H)$ of \CCHC, let $G' \gets \Transform(V,E,F,H)$. For any $uv\in F$, it holds that $u,v$ are connected in $G'$ and their neighborhoods are the same.
\end{lemma}
\begin{proof}
The idea is that all nodes in the same supernode are explicitly connected by the algorithm, in Line~\ref{line:connectSupernode}. Then all nodes of the same supernode connect to the exact same nodes due to the rounding step in Line~\ref{line:round}.

More formally, as $uv\in F$, they are trivially both in the same connected component of~$(V,F)$. Thus they are in the same supernode.

As $s(u)=s(v)$, $u$ and $v$ get connected in Line~\ref{line:connectSupernode}. All subsequent steps only modify edges~$u'v'$ where $s(u')\ne s(v')$, therefore $u,v$ remain connected in $G'$. Similarly both $u$ and~$v$ are connected with all other nodes in $s(u)$.

For nodes $w\not \in s(u)$, when $\set{s(u),s(w)}$ is processed in the loop of Line~\ref{line:round}, either both $u$ and $v$ get connected to $w$ or both get disconnected by $w$.
\end{proof}

Similarly, hostile nodes are disconnected and do not share any common neighbor.

\begin{lemma} \label{lem:hostileStructure}
Given a satisfiable instance $(V,E,F,H)$ of \CCHC, let $G' \gets \Transform(V,E,F,H)$. For any $uv\in H$ it holds that $u,v$ are not connected in $G'$ and they have no common neighbor.
\end{lemma}
\begin{proof}
The idea is that all nodes in hostile supernodes are explicitly disconnected by the algorithm, in Line~\ref{line:disconnectSupernode}. Then if two hostile nodes share a common neighbor, we drop both edges in Line~\ref{line:drop2}.

More formally, as the instance is satisfiable, we have that $s(u)\ne s(v)$ by Lemma~\ref{lem:impossible}. Therefore no node in $s(u)$ is connected with a node in $s(v)$ after Line~\ref{line:disconnectSupernode}. In Line~\ref{line:drop2} we only remove edges, meaning that when we process $\set{s(u),s(v)}$ in the loop of Line~\ref{line:round}, the two supernodes are not connected, and they stay like that. Thus, $u,v$ (and even $s(u),s(v)$) are not connected in $G'$.

After Line~\ref{line:drop2}, for any supernode $W$ we have that at least one from $s(u),s(v)$ are not connected with $W$, or else the loop would not terminate. Assume without loss of generality that $s(u)$ is not connected with $W$. Therefore, $s(u)$ is also not connected with $W$ after the loop of Line~\ref{line:round}, meaning that even if $v$ is connected with a node $w\in W$, $u$ is not as $s(u)$ is not connected with $s(w)=W$. This guarantees that they have no common neighbor.
\end{proof}

With these lemmas, we can conclude that a \textsc{PIVOT} algorithm on $G'$ gives a clustering that satisfies the hard constraints. This was already observed in~\cite{deterministicPivoting}; we include a short proof for intuition, as we also use this lemma in \cref{sec:nodeWeightedAppendix}.

\begin{lemma} \label{lem:pivoting}
Let $(V,E,F,H)$ be a satisfiable instance of \CCHC{} and $G'=(V,E')$ be a graph such that any two friendly nodes are connected and have the same neighborhood in $G'$, while hostile nodes are not connected and have no common neighbor in $G'$. Then applying a PIVOT algorithm on $G'$ gives a clustering that satisfies the hard constraints. In particular, this holds for $G'=\Transform(V,E,F,H)$.
\end{lemma}
\begin{proof}
The idea is that due to the assumptions, the choice of the first pivot does not violate any hard constraint. As \textsc{PIVOT} algorithms progress, they work with induced subgraphs of the original graph, which also satisfy the assumptions, and therefore no hard constraint is ever violated.

For the sake of contradiction, assume that two hostile nodes $u,v$ are placed in the same cluster by a \textsc{PIVOT} algorithm. By definition of a \textsc{PIVOT} algorithm, this happens when we work with some $V'\subseteq V$ on the induced subgraph $G'[V']$, and we pivot on a node $w$ that is connected with both $u,v$. As $w$ is connected with both $u,v$ in $G'[V']$, it is also connected with $u,v$ in $G'$. But this contradicts the assumption on hostile nodes.

Similarly, for the sake of contradiction assume that two friendly nodes $u,v$ are put in separate clusters by a \textsc{PIVOT} algorithm. Without loss of generality assume that $u$ is the first to be placed in a cluster that does not contain $v$. Again, this happens when we work with some $V'\subseteq V$ on the induced subgraph $G'[V']$, and we pivot on a node $w$ that is connected with $u$ but not with $v$. As $G'[V']$ is an induced subgraph, $w$ is connected with $u$ but not with $v$ in $G'$.
But this contradicts the assumption on friendly nodes.

Therefore, by Lemmas~\ref{lem:friendlyStructure} and~\ref{lem:hostileStructure} the claim holds for $G'=\Transform(V,E,F,H)$.
\end{proof}

\section{\NWCC{}} \label{sec:nodeWeightedAppendix}

\subparagraph*{Deterministic Algorithm.}
We first give the deterministic \textsc{PIVOT} algorithms for Node-Weighted Correlation Clustering (\cref{thm:node-weighted-deterministic}). We summarize the pseudocode in \cref{alg:node-weighted-deterministic}. The analysis of the deterministic algorithm is similar to the \textsc{PIVOT} algorithm in \cref{sec:deterministicPivoting}.

\begin{figure}[h]
\caption{\,The LP relaxation for Node-Weighted Correlation Clustering.} \label{fig:lp-node-weighted}
\vskip -1.5ex\rule{\linewidth}{.5pt}
\begin{mini}|s|<b>{}{\sum_{uv \in \binom V2} x_{uv}}{}{}
    \addConstraint{\frac{x_{uv}}{\omega_u \omega_v} + \frac{x_{vw}}{\omega_v \omega_w} + \frac{x_{wu}}{\omega_w \omega_u}}{\geq 1\quad}{\forall uvw \in T(G)}
    \addConstraint{x_{uv}}{\geq 0\quad}{\forall uv \in \textstyle\binom V2}
\end{mini}
\rule{\linewidth}{.5pt}
\end{figure}

\begin{algorithm2e}[h]
Compute a $(1 + \frac\epsilon3)$-approximate solution \raisebox{0pt}[0pt][0pt]{$x = \set{x_{uv}}_{uv \in \binom V2}$} of the LP in \cref{fig:lp-node-weighted} \label{alg:node-weighted-deterministic:line:LP}\;
$C \gets \emptyset$\;
\While{$V \neq \emptyset$}{
    Pick a pivot node $u \in V$ minimizing
    \begin{equation*}
        \frac{\displaystyle\sum_{vw : uvw \in T(G)} \omega_v \omega_w}{\displaystyle\sum_{vw : uvw \in T(G)} x_{vw}}
    \end{equation*} \label{alg:node-weighted-deterministic:line:ratio}\;
    \vspace*{-3.5ex}
    Add a cluster containing $u$ and all its neighbors to $C$\;
    Remove $u$, its neighbors and all their incident edges from $G$\;
}
\Return~$C$
\caption{The adapted \textsc{PIVOT} algorithm to $(3 + \epsilon)$-approximate Node-Weighted Correlation Clustering.} \label{alg:node-weighted-deterministic}
\end{algorithm2e}

\begin{lemma}[Correctness of \cref{alg:node-weighted-deterministic}] \label{lem:node-weighted-deterministic-correctness}
\cref{alg:node-weighted-deterministic} correctly approximates Node-Weighted Correlation Clustering with approximation factor $3 + \epsilon$.
\end{lemma}
\begin{proof}
We use the same notation as in \cref{lem:cluster-correctness}. That is, let $G^{(i)} = (V^{(i)}, E^{(i)})$ denote the graph $G$ after removing the $i$-th cluster and let $u_i$ denote the $i$-th pivot node. By the same reasoning as in \cref{lem:cluster-correctness}, the total cost of the node-weighted clustering constructed by the algorithm is exactly
\begin{equation*}
    \sum_{i=0}^{t-1} \sum_{\substack{vw:\\u_ivw \in T(G^{(i)})}} \omega_v \omega_w.
\end{equation*}
In order to bound this cost, we again bound the cost of selecting the average node $u$ as a pivot---this time however, we weight the nodes proportional to their weights $\omega_u$:
\begin{align*}
    \sum_{u \in V^{(i)}} \omega_u \cdot \sum_{\substack{vw:\\uvw \in T(G^{(i)})}} \omega_v \omega_w
    &= 3 \cdot \sum_{\substack{uvw \in T(G^{(i)})}} \omega_u \omega_v \omega_w \\
    &\leq 3 \cdot \sum_{\substack{uvw \in T(G^{(i)})}} \omega_u x_{vw} + \omega_v x_{wu} + \omega_w x_{uv} \\
    &= 3 \cdot \sum_{u \in V^{(i)}} \omega_u \cdot \sum_{\substack{vw:\\uvw \in T(G^{(i)})}} x_{vw}.
\end{align*}
In the second step, we applied the LP constraint. Using this inequality, we conclude that for any pivot node the ratio in Line~\ref{alg:node-weighted-deterministic:line:ratio} is bounded by~$3$. Assuming that $x$ is a $(1 + \frac\epsilon3)$-approximate solution to the LP in \cref{fig:lp-node-weighted}, we obtain the following upper bound on the cost of the node-weighted correlation clustering:
\begin{align*}
    \sum_{i=0}^{t-1} \sum_{\substack{vw\\u_ivw \in T(G^{(i)})}} \omega_v \omega_w
    &\leq 3 \cdot \sum_{i=0}^{t-1} \sum_{\substack{vw\\u_ivw \in T(G^{(i)})}} x_{vw} \\
    &\leq 3 \cdot \sum_{vw \in V} x_{vw} \\
    &\leq (3 + \epsilon) \cdot \OPT^{(\text{LP})} \\
    &\leq (3 + \epsilon) \cdot \OPT^{(\text{NWCC})}. 
\end{align*}
Here, in order to bound $\OPT^{(\text{LP})} \leq \OPT^{(\text{NWCC})}$ (the optimal cost of the node-weighted correlation clustering), we argue that any node-weighted correlation clustering can be turned into a feasible solution of the LP in \cref{fig:lp-node-weighted}. Indeed, for any pair $uv$ whose preference is violated assign $x_{uv} = \omega_u \omega_v$ and for any pair $uv$ whose preference is respected assign~$x_{uv} = 0$.
Recalling that every bad triplet involves at least one pair whose preference was violated we conclude that all constraints of the LP are satisfied.
Thus \(x\) is a feasible solution and we have that $\OPT^{(\text{LP})} \leq \OPT^{(\text{NWCC})}$.
\end{proof}

\begin{proof}[Proof of \cref{thm:node-weighted-deterministic}]
We use \cref{alg:node-weighted-deterministic}. The correctness follows from the previous \cref{lem:node-weighted-deterministic-correctness}. To appropriately bound the running time, we use the same insight as in \cref{sec:deterministicPivoting}: Since the LP in \cref{fig:lp-node-weighted} is a covering LP, we can use the combinatorial algorithm in \cref{thm:covering-combinatorial} to approximate the LP in Line~\ref{alg:node-weighted-deterministic:line:LP} in time $\widetilde\Order(n^3 \epsilon^{-3})$. This proves the first part of the theorem.

For the second part we instead use an all-purpose LP solver to find an exact solution to the LP in Line~\ref{alg:node-weighted-deterministic:line:LP} in time $\Order((n^3)^{2.373}) = \Order(n^{7.119})$~\cite{CohenLS19,Brand20,AlmanW21}.
\end{proof}

\subparagraph*{Randomized Algorithm.}
In this section we describe our optimal randomized \textsc{PIVOT} algorithm for Node-Weighted Correlation Clustering (\cref{thm:node-weighted-randomized}). As the decisive ingredient, we provide a data structure to perform weighted sampling on a decremental set:

\begin{lemma}[Weighted Sampling] \label{lem:weightedSampling}
Let $A$ be a set of initially $n$ objects with associated weights $\set{\omega_a}_{a \in A}$. There is a data structure supporting the following operations on $A$:
\smallskip
\begin{itemize}
\item {\normalfont$\textsc{Sample}()$:} Samples and removes an element $a \in A$, where $a \in A$ is selected with probability~$\omega_a / \sum_{a' \in A} \omega_{a'}$. 
\item {\normalfont$\textsc{Remove}(a)$:} Removes $a$ from $A$. 
\end{itemize}
\smallskip
The total time to initialize the data structure and to run the previous operations until $A$ is empty is bounded by $\Order(n)$, with high probability $1 - \frac{1}{n^c}$, for any constant $c > 0$.
\end{lemma}

\begin{algorithm2e}[t]
Initialize the weighted sampling data structure on $V$ with weights $\set{\omega_u}_{u \in V}$\;
$C \gets \emptyset$\;
\While{$V \neq \emptyset$}{
    $u \gets \textsc{Sample}()$\;
    Add a cluster containing $u$ and all its neighbors to $C$\;
    Remove $u$, its neighbors and all their incident edges from $G$\;
    Run $\textsc{Remove}(u)$ and $\textsc{Remove}(v)$ for all neighbors $v$ of $u$\;
}
\Return~$C$
\caption{The randomized \textsc{PIVOT} algorithm computing an expected $3$-approximation of Node-Weighted Correlation Clustering.} \label{alg:node-weighted-randomized}
\end{algorithm2e}

We postpone the proof of \cref{lem:weightedSampling} for now and first analyze the randomized algorithm in \cref{alg:node-weighted-randomized}. It can be seen as a natural generalization of the sampling algorithm from~\cite{pivoting}.

\begin{lemma}[Correctness of \cref{alg:node-weighted-randomized}] \label{lem:node-weighted-randomized-correctness}
\cref{alg:node-weighted-randomized} correctly approximates Node-Weighted Correlation Clustering with expected approximation factor $3$.
\end{lemma}
\begin{proof}
We borrow the notation from \cref{lem:node-weighted-deterministic-correctness}, writing $G^{(i)} = (V^{(i)}, E^{(i)})$ for the graph after removing the $i$-th cluster and $u_i$ for the $i$-th pivot node. Assuming the correctness of the sampling data structure (\cref{lem:weightedSampling}), $u_i$ is sampled from $V_i$ with probability $\omega_u / \sum_{v \in V} \omega_v$. Again, the total cost of the clustering computed by \cref{alg:node-weighted-randomized} is exactly
\begin{equation*}
    \sum_{i=0}^{t-1} \sum_{\substack{vw:\\u_ivw \in T(G^{(i)})}} \omega_v \omega_w.
\end{equation*}
Let $x = \set{x_{uv}}_{uv}$ denote an optimal solution to the LP in \cref{fig:lp-node-weighted}. (In contrast to the deterministic algorithm, here we do not compute the solution.) To bound the inner sum in expectation, we use the same computation as in \cref{lem:node-weighted-deterministic-correctness}, relying on the LP constraint:
\begin{align*}
    \Ex_{u \in V^{(i)}}\left(\sum_{\substack{vw:\\uvw \in T(G^{(i)})}} \omega_v \omega_w\right)
    &= \frac{1}{\sum_{v \in V^{(i)}} \omega_v} \cdot \sum_{u \in V^{(i)}} \omega_u \cdot \sum_{\substack{vw:\\uvw \in T(G^{(i)})}} \omega_v \omega_w \\
    &= \frac{3}{\sum_{v \in V^{(i)}} \omega_v} \cdot \sum_{\substack{uvw \in T(G^{(i)})}} \omega_u \omega_v \omega_w \\
    &\leq \frac{3}{\sum_{v \in V^{(i)}} \omega_v} \cdot \sum_{\substack{uvw \in T(G^{(i)})}} \omega_u x_{vw} + \omega_v x_{wu} + \omega_w x_{uv} \\
    &= 3 \cdot \Ex_{u \in V^{(i)}}\left(\sum_{\substack{vw:\\uvw \in T(G^{(i)})}} x_{vw} \right).
\end{align*}
It follows that the expected total cost is bounded by
\begin{align*}
    \Ex_{u_0, \dots, u_{t-1}} \left(\sum_{i=0}^{t-1} \sum_{\substack{vw:\\u_ivw \in T(G^{(i)})}} \omega_v \omega_w\right)
    &\leq 3 \cdot \Ex_{u_0, \dots, u_{t-1}} \left( \sum_{i=0}^{t-1} \sum_{\substack{vw:\\u_ivw \in T(G^{(i)})}} x_{vw} \right) \\
    &\leq 3 \cdot \sum_{vw \in \binom V2} x_{vw} \\
    &\leq 3 \cdot \OPT^{(\text{LP})} \\
    &\leq 3 \cdot \OPT^{(\text{NWCC})},
\end{align*}
where we used the same arguments as in \cref{lem:node-weighted-deterministic-correctness}. In particular, we used that for \emph{any} choice of pivot nodes every pair $vw$ appears in the sum at most once and we can therefore drop the expectation.
\end{proof}

\begin{proof}[Proof of \cref{thm:node-weighted-randomized}]
By \cref{lem:node-weighted-randomized-correctness}, \cref{alg:node-weighted-randomized} correctly approximates Node-Weighted Correlation Clustering within a factor $3$, in expectation.

To bound the running time of the algorithm, first recall that by \cref{lem:weightedSampling} the total time to initialize and sample from the weighted sampling data structure is $\Order(n)$ with high probability $1 - 1/\poly(n)$. The time to construct the clustering is $\Order(n + m)$ as any edge in the graph is touched exactly once.
\end{proof}

We remark that this randomized algorithm can in fact be seen as a reduction from Node-Weighted Correlation Clustering to Constrained Correlation Clustering: For any node~$u$, construct a supernode containing $\omega_u$ many new nodes (all of which are joined by friendliness constraints). Then there is a one-to-one correspondence between node-weighted and constrained correlation clusters of the same cost. The constructed instance's size is proportional to the sum of weights and can thus be much larger than the original instance's size; however we can use our weighted sampling data structure to efficiently sample from it anyways.

\subparagraph*{Weighted Sampling.}
We finally provide a proof of \cref{lem:weightedSampling}. It heavily relies on Walker's Alias Method~\cite{Walker74,alias}, which we summarize in the following theorem:

\begin{theorem}[Alias Method~{{{\cite{Walker74,alias}}}}]
Let $A$ be a set of $n$ objects with weights~\makebox{$\set{\omega_a}_{a \in A}$}. In time $\Order(n)$ we can preprocess $A$, and then sample $a \in A$ with probability~$\omega_a / \sum_{a' \in A} \omega_{a'}$ in constant time.
\end{theorem}

We start with the description of our data structure. Throughout, we partition the objects in~$A$ into buckets $B_1, \dots, B_\ell$, such that the $i$-th bucket $B_i$ contains objects with weights in~$[2^{i-1},2^i)$. Assuming that each weight is bounded by $\poly(n)$, the number of buckets is bounded by $\ell = O(\log{n})$. For every bucket $i = 1, \dots, \ell$, we maintain an \emph{initial weight} $p_i$ and an \emph{actual weight} $q_i$. Initially, we set $p_i, q_i \gets \sum_{s \in B_i} w(s_i)$. Moreover, using the Alias Method we preprocess in $O(\log{n})$ time the set of buckets~$\set{B_1, \dots, B_\ell}$ with their associated initial weights $p_i$. Additionally, using the Alias Method we preprocess each individual bucket $B_i$ with associated weights $\set{\omega_a}_{a \in B_i}$.

Next, we describe how to implement the supported operations:
\smallskip
\begin{itemize}
    \item {\normalfont$\textsc{Remove}(a)$:} Let $a$ be contained in the $i$-th bucket $B_i$. We remove $a$ from $B_i$, and update the actual weight $q_i \gets q_i - \omega_a$ (but not the initial weight $p_i$). If $q_i < p_i / 2$, then we recompute $p_i \gets \sum_{a \in B_i} \omega_a$ and we recompute the Alias Method both on $B_i$ as well as on the set of buckets (based on their initial weights $p_i$).
    \item {\normalfont$\textsc{Sample}()$:} We sample a bucket $i = 1, \dots, \ell$ with probability proportional to its initial weight $p_i$ using the Alias Method. With probability $q_i / p_i$ we \emph{accept} the bucket, otherwise we \emph{reject} and sample a new bucket.
    
    Next, using the preprocessed Alias Method we sample an element $a$ from $B_i$. If the element is no longer contained in $B_i$ (as it was removed in the meantime), we repeat and sample a new element. As soon as an element $a$ is found we report $a$ and call $\textsc{Remove}(a)$.
\end{itemize}
\smallskip
First, in \cref{lem:weighted-sampling-correctness} we argue for the correctness of the data structure. Then in \cref{lem:weighted-sampling-preprocessing,lem:weighted-sampling-remove,lem:weighted-sampling-sample} we analyze the total running time.

\begin{lemma} \label{lem:weighted-sampling-correctness}
{\normalfont$\textsc{Sample}()$} correctly returns an element $a$ with probability $\omega_a / \sum_{a' \in A} \omega_{a'}$.
\end{lemma}
\begin{proof}
The statement is proven in two steps. First, we show that every bucket $B_i$ is selected with probability $\sum_{a \in B_i} \omega_a / \sum_{a' \in A} \omega_{a'}$. Indeed, it is easy to check that we invariantly have $q_i = \sum_{a \in B_i} \omega_a$. Moreover, every bucket $B_i$ is sampled with probability~$p_i / \sum_{a' \in A} \omega_{a'}$ by the Alias Method. Since we accept every bucket with probability~$q_i / p_i$ (and resample otherwise), the probability of accepting $B_i$ is indeed
\begin{equation*}
    \frac{p_i}{\sum_{a' \in A} \omega_{a'}} \cdot \frac{q_i}{p_i} = \frac{q_i}{\sum_{a' \in A} \omega_{a'}} = \frac{\sum_{a \in B_i} \omega_a}{\sum_{a' \in A} \omega_{a'}}.
\end{equation*}

Second, assume that we accepted $B_i$ and continue sampling $a \in B_i$. Since we resample whenever a non-existing element is returned, each existing element $a \in B_i$ is sampled with probability exactly $\omega_a / q_i$. By combining both steps, we obtain the claim.
\end{proof}

\begin{lemma} \label{lem:weighted-sampling-preprocessing}
The preprocessing time is bounded by $\Order(n)$.
\end{lemma}
\begin{proof}
The computation of $p_i$ and $q_i$ takes time $\Order(n)$. Moreover, to initialize the Alias Method on the buckets takes time $\Order(\log n)$, and to initialize the Alias Method on an individual bucket $B_i$ takes time $\Order(|B_i|)$. Thus, the total preprocessing time is bounded by~$\Order(n + \sum_i |B_i|) = \Order(n)$.
\end{proof}

\begin{lemma} \label{lem:weighted-sampling-remove}
The total time of all {\normalfont$\textsc{Remove}(\cdot)$} operations is bounded by $\Order(n)$.
\end{lemma}
\begin{proof}
For a given element $a$, we can find its bucket $B_i$ and update the actual weight $q_i$ of that bucket in constant time. Since there are at most $n$ operations, this amounts to time~$\Order(n)$. It remains to bound the time to reconstruct the Alias Method data structures.

We are left to argue about the total time of rebuilding. We only rebuild the structure of a single bucket if its actual weight dropped in half since the last rebuilding. As the objects in a bucket have weights that are within a factor $2$, we have removed at least a fraction of $\frac{1}{3}$ objects in the bucket since the last rebuilding. Thus, the total number of reconstructions of a bucket $B_i$ is $O(\log{|B_i|}) = \Order(\log n)$ and the total time for these reconstructions is bounded by~$\sum_{k=0}^\infty (\frac23)^k |B_i| = \Order(|B_i|)$. Summing over all buckets, the total reconstruction time is bounded by $\Order(n)$.

We also rebuild the Alias Method structure on the set of buckets each time a bucket is rebuilt. Since there are only $\Order(\log n)$ buckets, each of which is rebuilt at most $O(\log n)$ times with running time $\Order(\log n)$, the total time for this part is bounded by $\Order(\log^3 n)$.
\end{proof}

\begin{lemma} \label{lem:weighted-sampling-sample}
The total time of all {\normalfont$\textsc{Sample}()$} operations is bounded by $\Order(n)$ with probability~\makebox{$1 - \frac1{n^c}$}, for any constant $c > 0$.
\end{lemma}
\begin{proof}
Consider a single execution of $\textsc{Sample}()$. Since at any point during the lifetime of the data structure we have that $q_i > \frac{p_i}2$, we accept the sampled bucket with probability at least $\frac12$. Moreover, by the same argument we find an existing element in that bucket with probability $\frac12$ as well. The number of repetitions of both of these random experiments can be modeled by a geometric random variable with constant expectation. Hence, using a standard concentration bound for the sum of independent geometric random variables~\cite{Janson18}, the total number repetitions across the at most $n$ executions of $\textsc{Sample}()$ is bounded by $\Order(n)$ with probability at most $1 - \frac{1}{n^c}$, for any constant $c > 0$. 
\end{proof}

In combination, \cref{lem:weighted-sampling-correctness,lem:weighted-sampling-preprocessing,lem:weighted-sampling-remove,lem:weighted-sampling-sample} prove the correctness of \cref{lem:weightedSampling}.

\subparagraph*{Non-Integer Weights.}
Throughout we assumed that the weights are integers, but what if the weights are instead rationals (or reals in an appropriate model of computation)? Our deterministic algorithm uses the weights only in the solution of the LP and is therefore unaffected by the change. We remark that our randomized algorithm can also be adapted to $3$-approximate \NWCC{} with the same running time even if rational weights were allowed by adapting the weighted sampling structure.

\end{document}